\documentclass[11pt]{article}

\usepackage{amssymb}
\usepackage{amscd}
\usepackage{amsmath}
\usepackage{amsfonts}
\usepackage{amsthm}
\usepackage{color}
\usepackage{enumerate}

\theoremstyle{plain}

\newtheorem{theo}{Theorem}[section] 
\newtheorem{prop}[theo]{Proposition}
\newtheorem{lemma}[theo]{Lemma}
\newtheorem{cor}[theo]{Corollary}

\theoremstyle{definition}

\newtheorem{rem}[theo]{Remark}

\newcommand{\R}{{\mathbb{R}}}
\newcommand{\N}{{\mathbb{N}}}
\newcommand{\C}{{\mathbb{C}}}
\newcommand{\Z}{{\mathbb{Z}}}

\newcommand{\be}{{\beta}}
\newcommand{\al}{{\alpha}}
\newcommand{\la}{{\lambda}}
\newcommand{\De}{{\Delta}}
\newcommand{\si}{{\sigma}}

\newcommand{\ka}{{\kappa}}

\newcommand{\om}{{\omega}}

\newcommand{\Ga}{{\Gamma}}

\newcommand{\ep}{\epsilon}

\newcommand{\Ci}{{\mathcal{C}}^{\infty}} 
\newcommand{\Cl}{\mathcal{C}}
\newcommand{\op}{\operatorname}
\newcommand{\con}{\overline}
\newcommand{\bigo}{\mathcal{O}} 
\newcommand{\Hilb}{\mathcal{H}}

\newcommand{\er}{\op{er}}

\usepackage{adjustbox}
\newcommand{\bigWedge}{\textstyle{\bigwedge}}

\begin{document}

\author{Laurent Charles and Benoit Estienne \footnote{B.E. was supported by Grants No. ANR-17-CE30-0013-01 and No. ANR-16-CE30-0025.}}

\title{Entanglement entropy and Berezin-Toeplitz operators}

\maketitle

\begin{abstract} 
We consider Berezin-Toeplitz operators on compact K\"ahler manifolds whose
symbols are characteristic functions. When the support of the
characteristic function has a smooth boundary, we prove a two-term Weyl
law, the second term being proportional to the Riemannian volume of the
boundary. As a consequence, we deduce the area law for the entanglement
entropy of integer quantum Hall states.  Another application is for the
determinantal processes with correlation kernel the Bergman kernels of a positive line bundle : we prove
that the number of points in a smooth domain is asymptotically normal. 
\end{abstract}

The area laws for entanglement entropy have been widely discussed in recent years in condensed matter and quantum field theories.  Typically, one considers a many-particle state and a geometric partition of the space in two sub-regions. The von Neumann entropy of the reduced state of a sub-region measures the degree of entanglement between the two regions. The area law states that this entanglement entropy is proportional to the volume of the boundary of the sub-region. These area laws have been verified in many different systems, such as harmonic and spin chains, strongly correlated fermionic systems and quantum field theories, see for instance the surveys \cite{Pe_Ei}, \cite{Sur_ACP}, \cite{Sur_AFOV}.  

So far, in the mathematical literature, only the case of the free Fermi gas has been considered \cite{Gioev}, \cite{HeLeSp}, \cite{LSS1}, \cite{LeSoSp}.
The goal of the present paper is to address the case of Integer Quantum Hall (IQH) states. Unlike the Fermi gas, these systems have a spectral gap above the ground state energy and the corresponding area law has no logarithmic correction. 

As it was observed in \cite{Gioev_Klich}, the area law for a Fermi gas is related to a conjecture by Widom on the correction terms in the Weyl law for pseudo-differential operators with singular symbols. Similarly, for IQH states, the area law amounts to estimating a weighted spectral average of a Berezin-Toeplitz operator $T$ whose symbol is the characteristic function of the sub-region. The theory of Berezin-Toeplitz operators is well-developed but here we have to face several difficulties: first, the symbol of the operator $T$ is singular not even continuous; second, the expectation we have to estimate is the trace of $f(T)$ with a function $f$ having a logarithmic singularity at $0$ and $1$ precisely where the concentration of eigenvalues is highest; third, the relevant term in this estimate is a correction to the leading order term.

The results we obtain are quite general, for any domain with a smooth boundary of a compact Riemann surface, or in higher dimension of a compact K\"ahler manifold. This generality was a surprise to us, because even if the physics papers provide a lot of numerical evidence and heuristic arguments, they treat only very particular integrable geometries with a lot of symmetries, which allows one to compute explicitly an eigenbasis for the reduced state \cite{RS1,RS2}. 

Besides the entanglement entropy estimate, our Weyl law have applications
to determinantal processes whose correlation kernel is the Bergman kernel of
a positive line bundle. In this context, very general results are already
known for linear statistics, cf. \cite{Be1}, \cite{Be2}.  
 We give new estimates by providing complete asymptotic expansions for the
variance and the higher cumulants of the number of points in a smooth
domain. In particular, the variance is proportional to the volume of the
boundary at first order, a variant of the area law. This leads after a convenient rescaling to new convergences to the
normal distribution. 
  Our result could also find applications in the context of random matrix theory, in particular the distribution of eigenvalues in a smooth domain of the complex plane for the complex Ginibre ensemble. 
To compare with the existing litterature, similar central limit theorems hold more generally for Coulomb gases \cite{LeSe}, \cite{Se_surveyICM}, our setting corresponding to the inverse temperature $\be =2$. But to our knowledge, all theses results are for smooth enough linear statistics, the case of characteristic functions being new.

Our proofs are based on Bergman kernel asymptotics. The technical part consists of estimating singular oscillatory integrals. To do this, we develop a generalisation of the Laplace method which is of independent interest.

Before we state our results precisely, let us mention that partial closed
results in the Bargmann space were already obtained by Oldfield
\cite{Oldfield}, using Weyl quantization. Bergman kernel techniques have
already been used with success for the study of quantum Hall effect,
cf. the survey \cite{Kl} for instance. In the context of Berezin-Toeplitz operators in K\"ahler manifolds, two recent papers on different but related subjects are \cite{ZeZh} on partial Bergman kernels and \cite{PEU} on the quantization of submanifolds. Finally, the appendix of \cite{Pol_oim} is devoted to multiplicative properties of the Toeplitz operators with a characteristic function symbol.

\section{Statements of the results} 

\subsection{Toeplitz operators with characteristic function symbol} \label{sec:toepl-oper-with}
Let $L \rightarrow M$ be a positive holomorphic Hermitian line bundle on a compact manifold $M$. 
For any positive integer $k$, let $\mathcal{H}_k$ be the space of holomorphic sections of $L^k$.
 For any function $f$ on $M$, the {\em Toeplitz operator} with multiplicator $f$ is the endomorphism  $T_{f,k}$ of $\mathcal{H}_k$ such that
$$ \langle T_{f,k} s ,t \rangle = \langle f s , t \rangle , \qquad \forall \; s,t \in \mathcal{H}_k.$$
Here the scalar product $\langle \cdot, \cdot \rangle$ of sections of $L^k$  is defined by integrating the pointwise scalar product against the Riemannian measure $\mu$, 
the Riemannian metric $g$ of $M$ being determined by the curvature $\Theta(L)$ of the Chern connection of $L$, that is $g (X,Y) = - i \Theta(L) ( X,jY)$ with $j$ the complex structure. 

The definition of the Toeplitz operator $T_{f,k}$  makes sense for any integrable function $f$. So in particular we can consider the characteristic function of a mesurable subset $A$ of $M$. We denote by $T_{A,k}$ the corresponding Toeplitz operator. 

All the results we will prove concern the {\em semi-classical limit} $k \rightarrow \infty$. The dimension of $\mathcal{H}_k$ is equal to first order
\begin{gather} \label{eq:dim_vol}
 \op{dim} \Hilb_k = \left(\frac{k}{2\pi} \right)^n \left( \op{vol} (M) + \bigo \left(k^{-1}\right) \right)
\end{gather}
where $n$ is the complex dimension of $M$ and $\op{vol} (M)$ is the Riemannian volume of $M$.

For the reader not familiar with this geometric setting, it can be interesting to consider the example of the projective space $\mathbb{P}^N$ with $L_{\mathbb{P}^N}$ the dual of the tautological line bundle. More generally, let $M$ be a closed complex submanifold of $\mathbb{P}^N $ and $L$ be the restriction of $L_{\mathbb{P}^N}$ to $M$. In this case, the spaces $\mathcal{H}_k$ have a simple concrete description. 
Let $\pi$ be the projection from the unit sphere of $\C^{N+1}$ onto $\mathbb{P}^N$ and let $\tilde M := \pi^{-1} (M)$. Then when $k$ is sufficiently large the holomorphic sections of $L^k$ lifts to polynomial functions of $\C^{N+1}$ and we have a natural identification 
$$\mathcal{H}_k  \simeq \C_k [\tilde{M}],$$ 
where  $\C_k [\tilde {M}]$ consists of the restriction to $\tilde{M}$ of homogeneous polynomials of $\C^{N+1}$ with degree $k$. Furthermore, the scalar product is given in $\C_{k} [\tilde{M}]$ by $ \langle f, g \rangle = C_{n,N} \int_{\tilde M} f(x) \con{g} (x) \; d \tilde{\mu} (x)$ where $\tilde{\mu}$ is the Riemannian measure of $\tilde{M}$ for the Riemannian metric of $\tilde{M}$ induced by the Euclidean scalar product of $\C^{N+1}$ and $C_{n,N}$ is a positive constant independent of $k$ and $M$.

Let us return to the general situation and consider a subset $A$ of $M$ with an empty boundary. Then $M$ is the disjoint union of two open sets: $A$ and its complementary set $A^c$. So $\mathcal{H}_{k} = \mathcal{H}_k(A) \oplus \mathcal{H}_k (A^c)$ where $\mathcal{H}_k(A)$ (resp. $\mathcal{H}_k(A^c)$) consists of the sections vanishing on $A^c$ (resp. on $A$). Furthermore $T_{A,k}$ is the projector onto $\mathcal{H}_k(A)$ with kernel $\mathcal{H}_{k}(A^c)$. 
This property does not hold in general. Actually, if $M$ is connected and $A$ and $A^c$ have non empty interiors, then $0$ and $1$ are {not even} in the spectrum of $T_{A,k}$. Nevertheless, $T_{A,k}$ looks like a projector in the sense that its eigenvalues concentrate at $0$ and $1$.  Let us explain this more precisely. 

First of all, for any measurable set $A$ of $M$, $T_{A,k}$ is a Hermitian endomorphism, whose spectrum $\op{sp}(T_{A,k})$ is contained in $[0,1]$.  Then it is proved in \cite{Berndtsson}, \cite{Lin}  that if the boundary of $A$ has a null measure, then for any $\ep \in (0,1/2)$ we have 
\begin{gather} \label{eq:est_large}
\begin{split}
 \bigl| \op{sp} (T_{A,k}) \cap [0,\epsilon]  \bigr| = \Bigl(\frac{k}{2\pi} \Bigr)^n \bigl( \mu (A^c) + \op{o} (1) \bigr) \\ 
\Bigl| \op{sp} (T_{A,k}) \cap [1-\epsilon,1]  \bigr| = \Bigl(\frac{k}{2\pi} \Bigr)^n \bigl( \mu (A)+ \op{o} (1) \bigr)  
\end{split}
\end{gather}
in the limit $k \rightarrow \infty$, where the eigenvalues are counted with multiplicity.  Here and in the sequel we denote by $\mu (B) = \int_B \mu$ the measure of a subset $B$ of $M$. 

The next question is about the possible eigenvalues in $[\ep, 1-\ep]$. As a consequence of (\ref{eq:dim_vol}) and (\ref{eq:est_large}), we have   
 $\bigl| \op{sp} (T_{A,k}) \cap [\ep,1- \epsilon]  \bigr| = \op{o} (k^{n})$. Also it is not difficult to see that the corresponding eigenstates are concentrated on the boundary of $A$. More precisely, let $\mathcal{H}_k (\ep)$ be the sum of the eigenspaces of $T_{A,k}$ with eigenvalue in $[\ep, 1-\ep]$. Then for any compact subset $K$ of $M$ not intersecting $\partial A$, for any $N>0$, there exists $C_{N,K}$ such that for any $s \in \mathcal{H}_k(\ep)$ and $x \in K$, we have $$| s(x) | \leqslant C_{K,N} k^{-N} \|s \|.$$ So we can consider that the states in $\mathcal{H}_k(\ep)$ live in the interface of $A$ and its complementary set. 
 
 Our first result shows that the number of eigenvalues in $[\ep, 1- \ep]$ has a remarkable simple universal behavior when the boundary of $A$ is smooth. For any closed submanifold $B$ of $M$, denote by $\op{vol} ( B)$ the Riemannian volume of $B$. 

\begin{theo} \label{theo:intro}
Let $A$ be a subset of $M$ with a smooth boundary. Then for any $a,b$ such that $0<a<b<1$, we have 
\begin{gather} \label{eq:comptage}
\Bigl| \op{sp} (T_{A,k}) \cap [a,b] \Bigr| = k^{n-1/2} \frac{\op{vol} (\partial A)}{ ( 2 \pi )^n }   ( \er^{-1} (b) - \er^{-1} (a) ) + \op{o} ( k^{n-1/2})
\end{gather}
where $\er : \R \rightarrow [0,1]$ is the function defined by $\er(x) = \pi^{-1/2} \int_{-\infty}^x e^{-t^2} \; dt $.
\end{theo} 
The function $\er$ is increasing from $0$ to $1$, it is given in terms of the standard error function $\op{erfc}$ by $\er (x) = 1 - \op{erfc}(x) /2$. The transition between Theorem \ref{theo:intro} and estimates (\ref{eq:est_large}) is rather subtle
because $k^n$ is replaced by $k^{n-1/2}$ and the volumes of $A$ and $A^c$ are replaced by the volume of the boundary. So not surprisingly, the quantity $ | \er^{-1} (b) - \er^{-1} (a)|$ diverges as $a \rightarrow 0$ or $b \rightarrow 1$, which prevents us to reach (\ref{eq:est_large}) from Theorem \ref{theo:intro}. Nevertheless, we can improve Theorem \ref{theo:intro} and say something on the limits $a \rightarrow 0$, $b \rightarrow 1$ by considering weighted sums of eigenvalues as follows. 
Introduce the measure $m$ of $(0,1)$ such that $m ([a,b]) =  | \er^{-1} (b) - \er^{-1} (a) |$ and the corresponding integral  
$$ I(f) := \int_0^1 f(x) \; dm (x) = \int_{-\infty}^{\infty} f(\er (x) ) \; dx$$
for functions $f: (0,1) \rightarrow \C$. We have $ d m (t) = \delta (t) \; dt$ with a smooth density $\delta : ( 0,1) \rightarrow \R$ such that $\delta ( t) = \delta ( 1-t)$ and $$\delta ( t) \sim \frac{1}{2} t^{-1} \bigl ( \ln ( 1/t) \bigr) ^{-1/2}$$ as $t \rightarrow 0$. Consequently, the integral $I(f)$ converges for any continuous function $f: [0,1] \rightarrow \C$ which is H\"older continuous at $0$ and $1$ with  $f(0) = f(1)=0$.

\begin{theo} \label{theo:intro_2}
Let $A$ be a subset of $M$ with a smooth boundary.
For any continuous  $f : [0,1] \rightarrow \C$ satisfying $|f(t) | = \bigo ( t^p)$  and $|f(1-t) | = \bigo ( t^p)$ for some positive $p$, we have  
\begin{gather} \label{dev:f_trace}
 \op{tr} ( f ( T_{A,k}))  = k^{n-1/2} \frac{\op{vol} (\partial A)}{ ( 2 \pi)^n} I (f) \bigl( 1 + \op{o} (1) \bigr).
\end{gather}
If $f$ is polynomial with $f(0) = f(1) =0$, then we have a complete asymptotic expansion 
\begin{gather} \label{eq:expansion} 
 \op{tr}(f(T_{A,k})) = k^{n-1/2} \sum_{\ell = 0 }^{N} c_{\ell} (f) k^{- \ell/2 } + \bigo ( k^{n-N-1}) , \qquad \forall N \in \N 
\end{gather}
where $c_0 ( f) = (2 \pi )^{-n} \op{vol} (\partial A) I(f)$.  
\end{theo}

It is likewise that the expansion (\ref{eq:expansion}) holds also for
smooth functions but our proof works only for polynomial functions. 
%We will also show that for $g(x) = f(1-x)$, we have  $c_{\ell} (g) = (-1) ^{\ell} c_{\ell} (f) $. Hence, if $f(x) = f(1-x)$ (resp. $f(x) = - f(1-x)$), then $c_{\ell} (f) = 0$ for odd (resp. even) $\ell$. 
It is important that we prove (\ref{dev:f_trace}) under the H\"older assumption, because for our application to entanglement entropy, we will consider the function $f(x) = -x \ln x - (1-x) \ln (1-x)$. 

Estimate (\ref{dev:f_trace}) may be viewed as a subprincipal term in the Weyl law for $T_{A,k}$, as stated in the following result. 
\begin{cor} \label{cor:Weyl_law_2_terms}
For any continuous function $g: [0,1] \rightarrow \C$ which is H\"older-continuous at $0$ and $1$, we have 
\begin{xalignat*}{2} \op{tr} ( g ( T_{A,k}))  =  \Bigl( \frac{k}{2\pi} \Bigr)^n \bigl(  & g(0) \mu (A^c) + g(1) \mu (A)  +  k^{-1/2} \op{vol} (\partial A) I(f)  \\ & + \op{o} ( k^{-1/2}) \bigr) 
\end{xalignat*}
where $f(x) = g(x) - g(0)(1-x) - g(1) x$. 
\end{cor}

\begin{rem} 
The previous results hold as well in the Bargmann space $\mathcal{B}_{\hbar}$. Recall that for any $\hbar >0$, 
 $\mathcal{B}_{\hbar}$ is the subspace of $L^2 ( \C^n, e^{-\hbar^{-1} |z|^2}\mu )$ consisting of holomorphic functions. Here $\mu$ is the Lebesque measure of $\C^n$. For any measurable subset $A$ of $\C^n$, the Toeplitz operator $T_{A}(\hbar) $ is the bounded operator of $\mathcal{B}_\hbar$ defined by the same formula $\langle T_{A}(\hbar) s,t \rangle = \langle 1_A s, t \rangle$ as previously. Assume that $A$ is bounded, so that $T_{A}(\hbar)$ is a trace class operator. Then setting $\hbar = k^{-1} $, the second estimate of (\ref{eq:est_large}), Theorem \ref{theo:intro} and  Theorem \ref{theo:intro_2} hold. The proof is essentially the same.

In this particular case, the results are not completely new: in \cite{Oldfield}, Oldfield proved Theorem \ref{theo:intro} and Equation (\ref{dev:f_trace}) of Theorem \ref{theo:intro_2} for smooth $f$.  Previously, De Mari, Feichtinger and Nowak had obtained lower and upper bounds for $\op{tr}  (f ( T_{A} ( \hbar)))$ in \cite{MaFeNo}. Actually, the volume of $\partial A$ does not appear explicitly in the work of Oldfield, but Equations (\ref{eq:comptage}) and (\ref{dev:f_trace}) with smooth $f$ can be deduced from it. However, (\ref{dev:f_trace}) with a H\"older continuous $f$ does not follows from \cite{Oldfield}, neither the expansion (\ref{eq:expansion}).  The proof of \cite{Oldfield} is based on Weyl quantization, an approach which is not available in our geometric setting.  Instead we will use Bergman kernel techniques.
\qed \end{rem}

\begin{rem} Similar results to Theorem \ref{theo:intro} are known for
  truncated Toeplitz matrices. Let $V_k$ be the subspace of $L^2 (S^1)$
  spanned by the Fourier modes $e_0, \ldots, e_k$ where $e_\ell
  (x) = e^{ i \ell x}$. Let $A \subset S^1$ be a finite union of
  intervals and consider the endomorphism $M_k(A)$ of $V_k$ given by
  $\langle M_k (A) s ,t \rangle = \langle 1_A s, t \rangle$ for any $s,t
  \in V_k$. Then by \cite{Ba}, \cite{BaWi}, \cite{LaWi}, for any $0<a<b<1$, we have
\begin{gather*} \Bigl| \op{sp} ( M_k (A)) \cap [0,a] \Bigr| = k \frac{ \ell (S^1 \setminus A) }{ 2
  \pi} + \bigo ( \ln k ) , \\ \Bigl| \op{sp} ( M_k (A)) \cap [b,1] \Bigr| = k \frac{ \ell ( A) }{ 2
  \pi} + \bigo ( \ln k ) 
\end{gather*} 
where we denote by $\ell (B)$ the length of a subset $B$ of $S^1$. Furthermore 
$$ \Bigl| \op{sp} (M_k (A) \cap [a,b] \Bigr| \sim \ln k \frac{ | \partial A| }{ 2
  \pi^2} ( m^{-1} (b) - m^{-1} (a) ) 
$$ where $|\partial A|$ is the number of boundary points of $A$ and $m(x) =
\frac{1}{2} (  1+ \tanh (x/2)) $. To compare with (\ref{eq:comptage}), the
$k^{-1/2}$ is replaced by $\ln k$  and the function $\er$ by
the function $m$. A generalisation of this result is known for
Wiener-Hopf operators, with the same logarithm and function $m$. This is the
so-called Widom conjecture, proved in dimension 1 by Widom itself
\cite{Widom} and in
higher dimension by Sobolev \cite{Sob}.
\qed \end{rem}

Let us sketch the main step of our proof. We first show the asymptotic expansion (\ref{eq:expansion}) for polynomial $f$. To do this we use that for any positive integer $p$
\begin{xalignat}{2} \label{eq:12}
\begin{split}  
 \op{tr} (T_{A,k}^p - T_{A,k}^{p+1} ) =  & \int_{ A^p \times A^c }  \Pi_k(x_1,x_2) \Pi_k( x_2, x_3) \ldots \Pi_k ( x_p, x_{p+1}) \\ & \Pi_k (x_{p+1}, x_1) \; \mu ( x_1) \ldots \mu ( x_{p+1}) 
\end{split}
\end{xalignat}
where $\Pi_k$ is the Bergman kernel of $L^k$. The asymptotic behavior of the Bergman kernel is well-known \cite{Ze}, \cite{MaMa}, \cite{oim_op}. From this we estimate the integral (\ref{eq:12}) by adapting the Laplace method. Typically, we have to handle integrals of the form 
\begin{gather} \label{eq:int_exemple}  
I(k) = \int_D e^{ - k ( t^{2}/2 + t^3 f(t,s) ) }  g(t,s) dt \; ds 
\end{gather}
where $f$ and $g$ are smooth functions, $g$ being compactly supported.
When $D =\R^2$, this is easily done by viewing $s$ as a parameter and applying the standard Laplace method. We get in this case $I(k) = k^{-1/2} \sum_{\ell \in \N} a_{\ell} k^{-\ell} $. We actually have a different kind of domain $D$ which is conic and contained in $\{ |s | \leqslant C|t| \}$ for some $C>0$. In this case,  $$I(k) = k^{ -1}  \sum_{\ell \in \N} b_{\ell} k^{-\ell/2} ,$$ 
so there are two important differences: even if the phase is degenerate only in the $t$ direction, the variable $s$ has to be considered as non-degenerate in the sense that $I(k)$ is of order $k^{-1}$. Furthermore, the expansion of $I(k)$ is in power of $k^{-1/2}$. We will develop a  Laplace method for generalisations of the integral (\ref{eq:int_exemple}) in any dimension and deduce the asymptotic expansion of  (\ref{eq:12}).

Once this is done, we can show the estimate (\ref{dev:f_trace}) for smooth $f$ and Theorem \ref{theo:intro}. To prove (\ref{dev:f_trace}) for H\"older-continuous $f$, we need the additional estimate: 
$$\op{tr}\bigl(  T^p_{A,k}(1-T_{A,k})^p \bigr) \leqslant C_p k^{n-1/2}$$ 
for any $p>0$, that we deduce from a Berezin-Lieb inequality.   

The proof in the Bargmann space requires a little extra work because $\C^n$ is not bounded, but it is also simpler because the Bergman kernel is given by an explicit formula.  

\subsection{Application to IQH states} \label{sec:appl-iqh-stat}

\subsubsection*{Free fermions}

Let $\mathcal{E}_k$ be the space of square integrable sections of $L^k$. It is a Hilbert space with the scalar product defined as previously by integrating the pointwise scalar product of sections against the Riemannian volume element. Let $d_k$ be the dimension of $\mathcal{H}_k$ and $(s_i, \;i =1, \ldots, d_k)$ be an orthonormal basis of $\mathcal{H}_k$. Define 
\begin{gather} \label{eq:fermionic_state} 
 \Psi := s_1 \wedge \ldots \wedge s_{d_k} \in {\bigWedge}^{d_k} \mathcal{E}_k .
\end{gather}
As we will explain at the end of this Section, $\mathcal{H}_k$ is the first Landau level of a magnetic Laplacian. So $\Psi$ represents a Fermionic state in which the first Landau level is fully occupied. $\Psi$ has norm $1$ for the natural scalar product of  $\bigwedge^{d_k} \mathcal{E}_k$. This space is not complete but it does not matter for what we will do. 

Since we work with square integrable sections instead of holomorphic sections, for any measurable set $A$ of $M$, we have the decomposition $\mathcal{E}_k = \mathcal{E}_k(A) \oplus \mathcal{E}_k(A^c)$, where $\mathcal{E}_k(A)$, $\mathcal{E}_k(A^c)$ are the image and kernel of the multiplication by the characteristic function $1_A$. Consequently
\begin{gather} \label{eq:dec_A_comp}
 {\bigWedge}^{d_k}  \mathcal{E}_k = \bigoplus_{\ell = 0 }^{d_k} \bigl( {\bigWedge} ^{\ell} \mathcal{E}_k (A) \bigr) \otimes \bigl( {\bigWedge}^{d_k - \ell} \mathcal{E}_k (A^c) \bigr) 
\end{gather}
Let $N_A$ be the endomorphism of $\bigwedge ^{d_k} \mathcal{E}_k$ acting by multiplication by $\ell$ in the $\ell$-th summand of (\ref{eq:dec_A_comp}). In statistical quantum mechanic, $N_A$ is the observable for the the number of particles in $A$. From the state $\Psi$ and the observable $N_A$, we obtain a probability distribution $N_A^\Psi$ defined by
\begin{gather} \label{def:N_A_Slater} 
 \Pr \, ( N_A^\Psi = \ell  ) = \| \Psi_{\ell} \|^2 , \qquad \ell =0 , \ldots , d_k 
\end{gather}
where $\Psi = \sum \Psi_\ell$ is the decomposition of $\Psi$ in the sum (\ref{eq:dec_A_comp}).  

In the mathematical litterature, the Slater determinants and their
associated distributions appeared as determinantal processes \cite{Be2}, \cite{AHM}
defined as follows. Identify first $\wedge^{d_{k}} \mathcal{H}_k$ with the
subspace of antisymmetric vectors of $\mathcal{H}^{\otimes d_k}$. Realize
$\mathcal{H}^{\otimes d_k}$ as the space of holomorphic sections of
$(L^k)^{\boxtimes d_k} \rightarrow M^{d_k}$. Then each normalized vector
$\Phi$ of $H^{0}(M^{d_k}, (L^{k})^{\boxtimes d_k} )$ defines a probability
measure $\mu_{\Phi}$ of $M^{d_{k}}$ given by $\mu_{\Phi}=   |\Phi|^2
\mu^{\boxtimes d_k}$. The measure $\mu_{\Psi}$ associated to $\Psi$ is a
determinantal process whose number of points in a given set $A$ has the same
distribution as $N_A^\Psi$.

The probability distribution $N_A^{\Psi}$ can be completely described in
terms of spectral invariants of the Toeplitz operator $T_A = T_{A,k}$. Let
$\la_{1} \leqslant \la_2 \leqslant \ldots \leqslant \la_{d_k}$ be the
eigenvalues of $T_A$. Recall that for $p\in [0,1]$, the Bernoulli random
variable $B(p)$ takes the value $1$ with probability $p$ and the value $0$
with probability $1-p$. Then $N_A ^{\Psi}$ has the same distribution has
the sum of independent random variables $B( \la_{i})$, $i =1 , \ldots ,
d_k$.  It is actually a general property of determinantal processes that
their number of points in a given set is a sum of Benoulli random variables, \cite{HKPV}.
This important fact appeared also implicitely in the physics litterature \cite{Klich},
\cite{PhysRevLett.100.086602}. We provide a short proof for our Fermionic states in section \ref{sec:quant-prob-ferm}.

An easy consequence is the computation of the cumulants of $N_A^\Psi$ as
spectral invariants of $T_A$. In particular, 
the expectation and variance of $ N_A^{\Psi}$ are given by 
$$ \mathbb{E} ( N_A^\Psi ) = \op{tr} (T_A), \qquad  \mathbb{E}\bigl( (N_A^\Psi)^2\bigr) - \bigl( \mathbb{E}(N_A^\Psi) \bigr)^2 = \op{tr} (T_A - T_A^2) .$$
More generally, the $\ell$-th cumulant is given by $\ka_{\ell} ( N_A^{\Psi} ) = \op{tr}  P_\ell ( T_{A})$
% \begin{gather} \label{eq:cum}
%  \ka_{\ell} ( N_A^{\Psi} ) = \op{tr}  P_\ell ( T_{A}), \qquad \ell \in \N^*
% \end{gather}
where the $P_{\ell}$ are defined recursively by $P_1 (X) = X$ and $P_{\ell+1} (X) =  X
(1-X)P'_{\ell} (X)$.
An application of Theorem \ref{theo:intro_2} gives the following asymptotics
expansion. 

\begin{theo} \label{theo:cumulant_estimate}
Assume that $A$ has a smooth boundary. Then we have the complete asymptotic expansions:
\begin{xalignat*}{2} 
&  \mathbb{E} ( N_A^\Psi )  = k^n  \Biggl( \frac{\mu(A)}{ (2 \pi)^n }  +  \sum_{m=1}^{N} \al_{m}(A) k^{-m} +  \bigo ( k^{-N-1} ) \Biggr),\\ 
&  \ka_{2\ell} ( N_A^\Psi )  = k^{n-1/2} \Biggl(\frac{\op{vol} (\partial A)}{(2\pi)^n} I(P_{2\ell}) +  \sum_{m=1}^{N} \be_{\ell,m}(A) k^{-m} +  \bigo ( k^{-N-1} ) \Biggr)    \\
&  \ka_{2\ell+1} ( N_A^\Psi )  =  k^{n-1}   \Biggl(\sum_{m=1}^{N} \eta_{\ell,m}(A) k^{-m} +  \bigo ( k^{-N-1} ) \Biggr)  .
\end{xalignat*}
for any $N$ and  $\ell \in \N^*$. 
\end{theo}
We will also provide estimates for the generating function of $N_A^\Psi$,
cf. Proposition \ref{prop:cumulant_generating_function}.  
By Theorem \ref{theo:cumulant_estimate}, the variance of $N_A^\Psi$ is
equal to 
\begin{gather} \label{eq:var_estim}
 \ka_2 (N_A^\Psi) =  k^{n-1/2} \frac{\op{vol}
(\partial A) }{ (2\pi)^{n}}I(P_{2}) + \bigo( k^{n-3/2} ).
\end{gather}
A consequence is the
following tail estimate and convergence to the normal distribution. Introduce the fluctuation
$\widetilde{N}_A^\Psi := N_A^\Psi - \mathbb{E} ( N_A ^{\Psi})$ and the
critical exponent $\al_c = n/2 - 1/4$.

\begin{cor} \label{cor:concentration+CLT} $ $ 
\begin{enumerate} 
\item there exists $k_0$ and $C>0$ such that for any $k \geqslant k_0$, for
  any $\beta>0$, we have $  \Pr \, \bigl(
  |\widetilde{N}_A^\psi | \geqslant k^{\al_c + \be}  \bigr) \leqslant e^{ - k^{\min (\al_c+ \be, 2 \be )} /C}. $
\item  $k^{-\al_c} \widetilde{N}_A^\psi$ converges in distribution to a centered normal random variable with variance $ (2\pi)^{-n} \op{vol} (\partial A)  I(P_{2})$.
\end{enumerate}
\end{cor}

The variance estimate (\ref{eq:var_estim})  can be compared with earlier results. Consider the
linear statistics of $\mu_{\Psi}$, that is the random variables of
$(M^{d_k}, \mu_{\Psi})$ of the form $\mathcal{N}[u](x_1, \ldots, x_{d_k}) =
\sum_{i=1}^{d_{k}} u(x_i)$ with $u:M \rightarrow \R$ a general
function. When $u$ is the characteristic function $1_A$, $\mathcal{N}[u]$
has the same distribution as $N_A^\Psi$. By  \cite[Theorem 1.5]{Be1}, when
$u$ is Lipschitz, the variance of $\mathcal{N}[u]$ is of order
$k^{n-1}$. So the regularity of $u$ brings a smaller variance.

% \begin{rem} 
%  We can deduce from Corollary \ref{cor:cumulant_estimate} the following asymptotics of the fluctuation $\widetilde{N}_A^\Psi = N_A^\Psi - \mathbb{E} ( N_A ^{\Psi})$. Introduce the critical exponent $\al_c = n/2 - 1/4$. Then we have the tail estimates: for any $\al > \al_c$ and $\epsilon >0$, 
% $$ \Pr \, \bigl( k^{-\al} |\widetilde{N}_A^\psi | \geqslant \ep \bigr) = \bigo ( k^{-\infty}). $$
% Whereas in the critical regime,  $k^{-\al_c} \widetilde{N}_A^\psi$ converges in distribution to a centered normal random variable with variance $ (2\pi)^{-n} \op{vol} (\partial A)  I(P_{2})$.
% \end{rem} 

The relevance of the estimates of the higher cumulants is not
obvious. Still, we can define a model probability distribution whose cumulants are the leading order terms appearing in Theorem \ref{theo:cumulant_estimate}. Let $\widetilde{B}(p) = B(p) - p$ be the fluctuation of the Bernoulli distribution. For any $n \in \N$ and $\al >0$, let $X_n (\al)$ be the sum of $2n+1$ independent random variables $\widetilde{B}( \op{er} ( \al m ) )$, $m =-n, -n +1, \ldots, n$. Then $X_n (\al)$ converges as $n \rightarrow \infty$ to a random variable $X(\al)$ with vanishing odd cumulants and even ones given by  
\begin{gather*} \kappa_{2 \ell} (X(\al)) = \sum_{m \in \Z} P_{2 \ell} ( \op{er} ( \al k) ) = \al^{-1} I ( P_{2 \ell}) + \bigo ( \al^{\infty}) . 
\end{gather*}
Choosing $\al^{-1}_k  =  k^{n-1/2} \op{vol} (\partial A)/(2\pi)^n$, we recover the leading order term of Theorem \ref{theo:cumulant_estimate}, that is for any $\ell \geqslant 2$, $$\kappa_{\ell}  ( N_A^{\Psi} ) = \kappa_{\ell} ( X(\al_k)) + \bigo (k^{n-1}). $$ An interpretation for this result is that the spectrum of $T_A$ can be approximated by the numbers $\op{er} ( \al_k m )$, $m \in \Z$ for what concerns the cumulant computations.

\subsubsection*{Entanglement entropy} 

Before we present our last application let us introduce the notion of entanglement, cf \cite{BLPY}, \cite{HaRa} for detailed presentations. Let $H$ be a complex Hilbert space that we assume finite dimensional to simplify the exposition. Denote by $\mathcal{L}(H)$ the space of self-adjoint endomorphisms of $H$ and by $\mathcal{S}(H)$ the subset of all positive trace 1 endomorphisms. The elements of $\mathcal{L}(H)$ are the {\em observables}, the elements of $\mathcal{S}(H)$ are the {\em mixed states}. To any observable $A$ and mixed state $\rho$ of $H$ is associated a probability distribution whose moments are $\op{tr} ( \rho ^n A)$, $n=1, 2, \ldots $.
  $\mathcal{S}(H)$ is a convex set whose extreme points are the rank one projectors, which are called the {\em pure state}. A nonzero vector $\phi$ of $H$ defines a pure state $P(\phi)$ which is the orthogonal projection onto $\C \phi$. Furthermore any mixed state $\rho$ may be considered as a classical statistical mixture of pure states. The von Neumann {\em entropy} of a mixed state $\rho \in \mathcal{S}(H)$ is defined by 
$$ S( \rho)  = - \op{tr} ( \rho \ln \rho )  \in [ 0,\infty)  .$$
It measures the amount of ``mixedness'' of a state. Typically, $S( \rho) =0$  if and only if $\rho$ is pure, whereas $S(\rho)$ is maximal for $\rho = ( \op{dim} H)^{-1} \op{id}_H$.

Consider now two finite dimensional Hilbert spaces $H_1$, $H_2$ and let $H = H_1 \otimes H_2$. Given a mixed state $\rho \in \mathcal{S}(H)$, the partial trace $\rho_1 := \op{tr}_{H_2} ( \rho) $ is a mixed state of $H_1$, called the {\em reduced} state. It is characterized amongst the mixed states of $H_1$ by  
$$ \op{tr} ( \rho_1 A_1 ) = \op{tr} \bigl( \rho ( A_1 \otimes \op{id}_{H_2} ) \bigr), \qquad \forall A_1 \in \mathcal{L} (H_1). $$
So $\rho_1$  represents the state $\rho$ in the subsystem $H_1$ since it allows to compute the expectation values of all the observables of $H_1$. In contrast to classical systems, it is possible that $\rho$ is pure and $\rho_1$ is not, so $\rho$ is not a product of pure states. For such {\em entangled} pure state, we define  the {\em entanglement entropy} as the von Neumann entropy of $\rho_1$. This quantity is a measure of the degree of entanglement of $\rho$. 

Back to our application, the Fermionic state $\Psi$ defined in (\ref{eq:fermionic_state}) may be viewed as a state of a bipartite system $\bigwedge \mathcal{E}_k = \bigl( \bigwedge \mathcal{E}_k(A) \bigr) \otimes \bigl( \bigwedge \mathcal{E}_k(A^c) \bigr) $. The reduced state $\rho_A $ of $\Psi$ is a finite rank endomorphism and its von Neumann entropy is given in terms of the Toeplitz operator $T_A$ by 
\begin{gather} \label{eq:entropy}
 S ( \rho_A) = \op{tr} \bigl( f (T_A)\bigr) \qquad \text{ with } \quad f (x) = -x \ln (x) - (1-x) \ln (1-x)
\end{gather}
This relation is generally deduced from Wick's Theorem \cite{Pe}. For the convenience of the reader, we will present an elementary alternative proof in Section \ref{sec:quant-prob-ferm}. Now we can deduce from Theorem \ref{theo:intro_2} the asymptotic behavior of $S( \rho_A)$. 

\begin{theo}
If $A$ has a smooth boundary, we have
$$ S( \rho_A) = k^{n-1/2} \frac{\op{vol}( \partial A) }{(2\pi)^n} I(f) ( 1 + \op{o} (1))$$ in the limit $k \rightarrow \infty$, where $f (x) = -x \ln x - (1-x) \ln (1-x)$. 
\end{theo}

\subsubsection*{Landau level} 

For an electron in a plane with an external perpendicular magnetic field $B$, the Hamiltonian is $H = \frac{1}{2} ( P_x^2 + P_y^2)$ with $P_x = \frac{1}{i} \partial_x + \frac{1}{2} B y$ and $P_y = \frac{ 1}{i} \partial_y - \frac{1}{2} B x $, cf. for instance \cite[Chapter 10]{Ezawa}. Here we have put all the physical constants equal to 1 except the perpendicular component $B$ of the magnetic field. Then it is customary to introduce the operators $a = \frac{1}{\sqrt{ 2B} } (P_x + i P_y)$, $a^* = \frac{1}{\sqrt{ 2B} } (P_x - i P_y)$ 
so that 
\begin{gather} \label{eq:13}
H = B \Bigl( a^* a + \frac{1}{2} \Bigr).
\end{gather}
So the lowest eigenvalue of $H$ is $B/2$ and the corresponding eigenspace, the lowest Landau level, consists of the $\psi$ satisfying $a \psi=0$. Introduce the complex coordinate $z =  (x + i y )/\sqrt 2$, then $( P_x + i P_y)/ \sqrt 2 = \frac{1}{i} ( \partial_{\con{z}} + B z )$. We conclude that $a \Psi =0$ if and only if $\psi = f e^{ - B|z|^2}$ with $f$ holomorphic. So the lowest Landau level is identified with the Bargmann space $\mathcal{B}_{\hbar}$ with $\hbar = B^{-1}$. 

In the geometric setting, starting from a Riemannian manifold $(M,g)$ and a Hermitian line bundle $L \rightarrow M$ equipped with a connection $\nabla$, we define the magnetic Laplacian 
$$H = \nabla^* \nabla$$ acting on $\Ga (M, L)$. 
The magnetic field $B$ is equal to $i$ multiplied by the curvature of $\nabla$. If $M = \R^2$ and $L$ is the trivial line bundle with connection $\nabla = d  + \frac{i}{2} B ( y dx - x dy )$, we recover the previous Hamiltonian.

Suppose now that $L$ is a positive holomorphic Hermitian line bundle $L \rightarrow M$, with the connection $\nabla$ being the Chern connection and the Riemannian metric $g$ being the metric determined by the curvature of $\nabla$ as explained in the beginning of Section \ref{sec:toepl-oper-with}. Then for any positive integer $k$, we have a magnetic Laplacian $H_k $ acting on sections of $L^k$. 

We have another natural Laplacian acting on sections of $L^k$ which is the holomorphic Laplacian $\Delta_k = \con \partial ^* \con{\partial}$. The Bochner-Kodaira formula relates these two Laplacians
$$ H_k = 2 \Delta_k + nk .$$
It is a generalization of (\ref{eq:13}). A proof may be found in \cite[7.3]{Dem2} or  \cite[Proposition 3.71]{BeGeVe}. 
By definition $\Delta_k$ is non-negative and  its kernel is the space $\Hilb_k$ of holomorphic sections of $L^k$. So the lowest eigenvalue of $H_k$ is $nk$, and the lowest Landau level is $\mathcal{H}_k$.  The magnetic field is $i \Theta (L^k) = i k \Theta (L)$. So the semiclassical limit $k \rightarrow \infty $ corresponds to a large magnetic field.

\subsection*{Acknowledgements}

The authors would like to express their warm gratitude to Beno\^it Dou\c{c}ot for bringing about this collaboration and for pointing out that the IQH entanglement entropy can be computed in terms of the spectrum of a convenient Toeplitz operator. B.E. also thanks Nicolas Regnault and Semyon Klevtsov for valuable discussions.

\section{K\"ahler quantization}

Let $L \rightarrow M$ be a holomorphic Hermitian line bundle on a complex compact manifold $M$. Recall that $L$ has a natural connection $\nabla$, called the Chern connection and   determined by the condition that it is compatible with the holomorphic and Hermitian structures.   If $\si$ is a local holomorphic frame of $L$, then $\nabla \si =- (\partial \varphi ) \otimes s$ where $\varphi = - 2 \ln |\si|$. The curvature of $\nabla$ is given by $\Theta (L) =  \partial \con{\partial} \varphi$. So $\om := i \Theta(L)$ is a real two form given in local complex coordinates by
$$ \om = i \frac{\partial^2 \varphi }{\partial z_i \partial \con{z}_j} \; dz_i \wedge d\con{z}_j .$$
We assume that $L$ is positive, meaning that $(\partial^2 \varphi/ \partial z_i \partial \con{z}_j)$ is a positive definite matrix at any point. So $\om$ is a K\"ahler form. We also have a Riemannian metric defined by $g(X,Y) = \om ( X, jY)$. The Riemannian volume density will be denoted by $\mu$. 

Let $A \rightarrow M$ be another holomorphic Hermitian line bundle not necessarily positive.  Since $M$ is compact, the space $\mathcal{H}_k$
of holomorphic sections of $L^k \otimes A$ is finite dimensional. Its dimension is given at first order in the limit $k \rightarrow \infty$ by 
\begin{gather} \label{eq:dim_est}
 \dim \mathcal{H}_k = \Bigl( \frac{k}{2 \pi } \Bigr)^n \Bigl( \mu (M) + \bigo(k^{-1}) \bigr) , 
\end{gather}
 where $\mu (M) = \int_M \mu$.  $\mathcal{H}_k$ has a natural scalar product 
\begin{gather} \label{eq:scalar_product}
 \langle s,t \rangle = \int_M (s,t) (x) \; d\mu (x) , \qquad s, t \in \Hilb_k
\end{gather}
where $(s,t)(x)$ denotes the pointwise scalar product at $x$. 
To any function $f$ of $M$ and $k \in \N$ is associated a Toepliz operator $T_f : \mathcal{H}_k \rightarrow \Hilb_k$ defined by
\begin{gather} \label{eq:Toeplitz}
 \langle T_f s, t \rangle = \langle fs, t\rangle, \qquad \forall s,t \in \mathcal{H}_k, 
\end{gather}
where the scalar product on the right-hand side is still defined by (\ref{eq:scalar_product}). Basic properties of Toeplitz operators are
\begin{gather} \label{eq:prop_Toeplitz_mult} 
 T_f T_g  = T_{fg} + \bigo ( k^{-1}) \\  \label{eq:prop_Toeplitz_trac}
\op{tr} T_f = \Bigl( \frac{k}{2 \pi } \Bigr)^n \int_M f \; d\mu  + \bigo ( k^{-1})
\end{gather}
for any smooth functions $f,g$, where the $\bigo$ depends on $f$ and $g$. Observe that the Toeplitz operator $T_f$ is well-defined by (\ref{eq:Toeplitz}) for any integrable  multiplicator $f$. The trace estimate (\ref{eq:prop_Toeplitz_trac}) still holds in this case, whereas (\ref{eq:prop_Toeplitz_mult}) is not true for integrable functions as we will see. 

The Bergman kernel of $\mathcal{H}_k$ is the holomorphic section $\Pi_k$ of $(L^k\otimes A)\boxtimes (\con{L}^k \otimes \con A)$ defined by 
$$ \Pi_k (x,y) = \sum_{i=1}^{d_k} s_{i,k} (x) \otimes \con{s}_{i,k} (y) $$
where $(s_{i,k})_{i =1 , \ldots, d_k}$ is any orthonormal basis of $\mathcal{H}_k$.
The restriction of the Bergman kernel to the diagonal can be considered as a function on $M$  because we have natural identification $L_x \otimes \con{L}_x \simeq \C$ and $A_x \otimes \con{A}_x \simeq \C$ given by the metrics of $L$ and $A$. We have 
\begin{gather} \label{eq:diag}
 \Pi_k (x,x) = \Bigl(\frac{k}{2 \pi } \Bigr)^n ( 1 + \bigo ( k^{-1}))
\end{gather} 
where the $\bigo$ is uniform in $x \in M$. 
Furthermore for any compact set $K$ of $M^2$ not intersecting the diagonal, for any $N$,  there exists $C_N$ such that 
\begin{gather} \label{eq:out_diag}
 |\Pi_k (x,y) | \leqslant C_N k^{-N} , \qquad \forall (x,y) \in K .
\end{gather}
The transition between (\ref{eq:diag}) and (\ref{eq:out_diag}) may be described as follows. Define a distance on $M$ by embedding $M$ into $\R^N$ and restricting the Euclidean distance. Denote by $|x-y|$ the distance between $x$ and $y$. Then  there exists constants $C >0$, $(C_N)_{N\in \N}$ such that for any $x, y\in M$ and $N\in \N$ we have
\begin{gather} \label{eq:decroissance} 
 \bigl| \Pi_k (x,y) \bigr| \leqslant Ck^{-n} e^{-k|x-y|^2/C} + C_N k^{-N}.
\end{gather}
We can actually characterize the Bergman kernel up to a $\bigo ( k^{-\infty})$ in terms of geometric datas as follows. For any tangent vector $X$ of $M$, we denote by $X^{1,0}$ and $X^{0,1}$ its holomorphic and antiholomorphic parts, so $X^{1,0} = \frac{1}{2} (X - ij X) $ and $X^{0,1} = \frac{1}{2}  (X+ ij X)$, where $j$ is the complex structure.

\begin{theo} \label{theo:bergman-kernel}
We have for any $k \in \N$
$$ \Pi_k (x,y) = \Bigl( \frac{k}{2 \pi } \Bigr)^n E^k (x,y) a(x,y, k) + R_k (x,y)$$
where 
\begin{enumerate} 
\item $E$ is a section of $L \boxtimes \con{L}$ independent of $k$. $|E|<1$ outside the diagonal, $E=1$ on the diagonal. Furthermore $\nabla E = \frac{1}{i} \al_E \otimes E$ on a neighborhood of the diagonal, where $\al_E$ is a one form of $M^2$ vanishing along the diagonal and such that for any vector fields $X_1$, $X_2$, $Y_1$ and $Y_2$ of $M$, we have 
\begin{gather} \label{eq:der_seconde_E} 
 \mathcal{L}_{(X_1,X_2)} \al_E (Y_1, Y_2) = \om ( X_1^{0,1} - X_2^{0,1} , Y_1) + \om ( X_1^{1,0} - X_2^{1,0}, Y_2).
\end{gather}
\item $a(\cdot, k)$ is a sequence of sections of $A \boxtimes \con{A}$ having an asymptotic expansion 
$$ a(\cdot, k) = a_0 + k^{-1} a_1 + \ldots + k^{-N} a_N + \bigo ( k^{-(N+1)}), \qquad \forall N$$
where the $\bigo$ is uniform on $M^2$, the coefficients $a_\ell$ are sections of $A \boxtimes \con{A}$, the restriction of $a_0$ to the diagonal being constant equal to $1$. 
\item $R_k$ is a section of $(L^k\otimes A)\boxtimes (\con{L}^k \otimes \con A)$ whose pointwise norm is uniformly in $\bigo(k^{-N})$ for any $N$. 
\end{enumerate}
\end{theo}

Estimates (\ref{eq:diag}), (\ref{eq:out_diag}) and (\ref{eq:decroissance}) are all consequences of Theorem \ref{theo:bergman-kernel}. In particular, to deduce (\ref{eq:decroissance}), observe that the conditions satisfied by $E$ imply that $\varphi_E = - 2 \ln |E|$ vanish along the diagonal and is positive outside the diagonal. A short computation from (\ref{eq:der_seconde_E}) gives the Hessian of $\varphi_E$ along the diagonal: for any vector fields $X_1$, $X_2$ of $M$, 
$$ \bigl( \mathcal{L}_{(X_1,0)} \mathcal{L}_{(X_2,0)} \varphi_E \bigr)(x,x)  = g\bigl( X_1,X_2\bigr)  (x).$$
where we denote by $\mathcal{L}$ the Lie derivative and by $g$ the Riemannian metric as above.

Theorem \ref{theo:bergman-kernel} has been deduced in \cite{oim_op} from the Szeg\"o kernel description of \cite{BoSj}. For other results on the Bergman kernel asymptotic behaviour, we refer the reader to \cite{Ze} and \cite{MaMa}.

\section{First spectral estimates} 

The Toeplitz quantization is real and positive in the sense that for any  integrable function $f$ of $M$, the corresponding Toeplitz operators $T_{f}$ is Hermitian when $f$ is real and non negative when $f$ is non negative. Since for any mesurable set $A$ of $M$, the characteristic function of $A$ satisfies $0 \leqslant 1_{A} \leqslant 1$, we deduce that $T_{A} := T_{1_A}$ is Hermitian with its spectrum in $[0,1]$. 

As mentioned in the introduction, if $ A \subset M$ has a non empty interior and $M$ is connected, then $0$ is not an eigenvalue of $T_{A}$. Indeed, if $T_A s =0$ then $\langle 1_A s, s \rangle =0$, which implies that $s =0$ on $A$, so $s =0$ by analytic continuation. Similarly, if the exterior of $A$ is non empty and $M$ connected, $1$ is not an eigenvalue of $T_A$.

In the sequel we will use the probability measure $\nu$ of $M$ obtained by renormalising $\mu$, so $\nu (B) = \mu (B) / \mu (M)$ for any subset $B$. We denote by $d_k$ the dimension of $\mathcal{H}_k$.

\begin{theo} \label{theo:prod_singular}
If $A$ and $B$ are two subsets of $M$ with null boundary measure, then $d_k^{-1} \op{tr}( T_A T_B ) = \nu (A \cap B) + \op{o}(1)$ as $k \rightarrow \infty$.
\end{theo}

When $A$ and $B$ are more regular, we can improve this result  \cite[Appendix]{Pol_oim}.

\begin{proof} By (\ref{eq:prop_Toeplitz_mult}) and (\ref{eq:prop_Toeplitz_trac}), for any two smooth functions $f$, $g$, we have 
\begin{gather} \label{eq:4}
d_k^{-1} \op{tr} (T_f T_g) = \int_M fg \; d\nu + \bigo(k^{-1}) 
\end{gather}
 Introduce the sequence $\nu_k$ of probability measures of $M^2$ such that 
$$\int_{M^2} f(x) g(y) d \nu_k (x,y) = d_k^{-1}\op{tr}( T_f T_g).$$ 
We can explicitly compute that $\nu_k =  d_{k}^{-1} |\Pi|^2  \mu \otimes \mu$. Since the functions of the form $h(x,y) = f(x) g(y)$ with $f$ and $g$ smooth, span a dense subset of $\mathcal{C}(M^2)$, (\ref{eq:4}) implies that $\nu_k$ converges weakly to the current of integration given by the measure $\nu$ on the diagonal $\Delta$, that is for any continuous function $h$ of $M^2$
\begin{gather} \label{eq:weak_convergence}
 \int_{M^2} h(x,y) \; d\nu_k (x,y) \rightarrow \int_{M} h(x,x) d\nu (x)
\end{gather}
By Portmanteau theorem \cite[Theorem 4.14.4]{Si}, this weak convergence implies that for any measurable subset $C$ of $M^2$ such that $\nu (  (\partial C) \cap \Delta) = 0$, we have $\nu_k (C) \rightarrow \nu ( C \cap \Delta)$. Applying this to $C = A \times B$ with $\nu(\partial A) = \nu ( \partial B) =0$, we get that $\nu_k (A\times B) \rightarrow  \nu (A\cap B)  $.  
\end{proof}

\begin{cor} If the boundary of $A$ is of measure zero, then for any $a,b$ such that $0<a<b<1$, we have
in the limit $k \rightarrow \infty$ 
\begin{xalignat*}{2}
  d_{k}^{-1} \bigl| \op{sp} (T_A) \cap [0, a] \bigr| & = \nu (A^c) + \op{o} (1) , \qquad d_{k}^{-1} \bigl| \op{sp} (T_A) \cap [a,b]\bigr| =  \op{o} (1), \\ 
 d_{k}^{-1} \bigl| \op{sp} (T_A) \cap [b, 1] \bigr| & = \nu (A) + \op{o} (1).
\end{xalignat*}
\end{cor}

This result was already proved in \cite{Berndtsson} by the same method. 
\begin{proof} Let $f(x) = 1 -x -a^{-1} x (1-x)$. Then $f(x) \leqslant 1_{[0,a]}(x)$ on $[0,1]$. So $$\op{tr} (f(T_A)) \leqslant \op{tr} (1_{[0,a]}(T_A)) = \bigl| \op{sp}(T_A) \cap [0,a] \bigr|. $$ 
By (\ref{eq:prop_Toeplitz_trac}), $d_k^{-1} \op{tr}( T_A) = \nu (A) + \bigo ( k^{-1})$. By Theorem \ref{theo:prod_singular},  $d_k^{-1} \op{tr} (T_A^2) = \nu(A) + \op{o}(1)$.
Hence, we have $d_k^{-1} \op{tr} (T_A - T_A^2) = \op{o}(1)$. If follows that $d_k^{-1} \op{tr} (f(T_A)) =  1 - \nu (A) + \op{o} (1)$. So
$$ d_k^{-1} \bigl| \op{sp}(T_A) \cap [0,a] \bigr| \geqslant 1 - \nu (A) + \op{o} (1) .$$
Applying exactly the same argument to $T_{A^c} = \op{id} - T_A$ instead of $T_A$ and $1-b$ instead of $a$, we get 
$$ d_k^{-1} \bigl| \op{sp}(T_A) \cap [1-b,1] \bigr| \geqslant \nu (A) + \op{o} (1) .$$
and the result follows.
\end{proof}
Introduce now the coherent state $e_x$ at $x \in M$ given by $e_x = \Pi_k (\cdot ,x)\in \mathcal{H}_k \otimes \con{L}^k_x \otimes \con{A}_x$. We are going to estimate the norm of $e_x$ on $A$, 
$$ \| e_x \|_A^2 := \int_A | e_x(y) |^2  \; d\mu (y) = \int_A | \Pi_k (x,y) |^2 \; d\mu (y).$$
When the boundary of $A$ is a smooth closed submanifold of $M$, we can introduce a {\em defining function} $\rho$ of $A$. It is a smooth real valued function on $M$ such that $\{ \rho < 0 \} = \op{int} A$, $\{ \rho = 0 \} = \partial A$ and $0$ is a regular value.

\begin{prop} \label{prop:loc_et_coh} $ $
\begin{enumerate} 
\item
For any compact subset $K$ of $M$ not intersecting the closure of $A$, for any $N$, there exists $C$ such that $\| e_x \|_A \leqslant C k^{-N}$ for any $x \in K$. 
\item If $\partial A$ is smooth and $\rho$ is a defining function of $A$, we have for any $N$
$$ \| e_x \|_A \leqslant Ck^{n/2} e^{ - k \rho^2(x) /C }   + C_Nk^{-N} $$
for any $x \notin \con{A}$ with $C$, $C_N$ independent of $x$.
\end{enumerate}
\end{prop}

\begin{proof} The first part follows directly from (\ref{eq:out_diag}). For the second part, we use (\ref{eq:decroissance}) and the fact that for any $x \notin A$ and $y \in A$, $\rho (x) \leqslant | \rho (x) - \rho (y) | \leqslant C |x -y |$ because $\rho$ is Lipschitz. So $ - |x-y|^2 \leqslant -\rho(x)^2/C^2$ and it follows that 
\begin{xalignat*}{2}
 \int_A e^{-k |x-y|^2/C' } \mu (y) & \leqslant  e^{ -k\rho (x)^{2} /2C^2 C'} \int_A e^{ - k | x-y|^2/ 2C'} \mu (y) \\ 
& \leqslant    C'' k^{-n}  e^{ -k\rho (x)^{2} /2C^2 C'}
\end{xalignat*}
by computing the integral in coordinates with a change of variable $x'= \sqrt{k} x$, which concludes the proof.
\end{proof}

Let us deduce the estimate for section $s$ in the sum $\Hilb_k (\epsilon)$ of the eigenspaces of $T_{A}$ with eigenvalue in $[\ep, 1-\ep]$. 
\begin{prop} $ $
\begin{enumerate} 
\item For any $\ep \in (0,1/2)$, $N \in \N$ and compact subset $K$ of $M$ not intersecting $\partial A$, there exists $C$ such that for any $s \in \Hilb_k ( \ep)$, we have 
$$ |s (x) | \leqslant C k^{-N} \| s\| , \qquad \forall \; x \in K, \; k \in \N.$$
\item Assume $A$ has a smooth boundary. Then for any $\ep \in (0,1/2)$ and $N \in \N$, there exists $C$ such that for any  $s$ in $\Hilb_k ( \ep)$ with norm 1, we have 
$$ |s(x) | \leqslant Ck^{n -1/4} e^{ - k \rho^2(x) /C }   + Ck^{-N}
$$ 
where $\rho$ is a defining function of $A$.
\end{enumerate}
\end{prop}

\begin{proof} We first prove that for any $s \in \Hilb_k ( \ep)$, for any $x \in M$,  
\begin{gather} \label{eq:7}
 |s (x) | \leqslant \ep^{-1} \| s \|  \bigl(\dim \mathcal{H}_k(\ep) \bigr)^{1/2} \| e_x \|_A 
\end{gather}
where $e_x$ is the coherent state at $x$. Let us write $s = \sum_{\la} c_\la s_\la$ where each $s_\la$ is a normalized eigensection of $T_A$ with eigenvalue $\la$. Then $ \| s \|^2 = \sum |c_{\la}|^2$, so  
$$  |s (x) |^2 \leqslant \| s \|^2 \sum_\la |s_{\la} (x) |^2 \leqslant \| s \|^2  \bigl(\dim \mathcal{H}_k(\ep) \bigr) \sup_{\la} |s_\la (x) |^2 $$
and  $ |s_{\la} (x) |  = \bigl| \langle s_\la, e_x \rangle \bigr| = \la ^{-1} \bigl| \langle T_A s_\la, e_x \rangle \bigr|  \leqslant \ep^{-1} \bigl| \langle s_\la , 1_A e_x \rangle \bigr| \leqslant \ep^{-1} \| e_x \|_A $
which shows (\ref{eq:7}). It is now easy to conclude by applying Proposition \ref{prop:loc_et_coh} to $T_A$ for $x$ in the exterior of $A$, and to $T_{A^c}$ for $x$ in the interior of $A$. Furthermore the dimension of $\mathcal{H}_k ( \ep)$ is a $\bigo ( k^{n-1/2})$ as will be proved later, cf Remark \ref{rem:nombre_vap_facile}. 
\end{proof}

\begin{theo} \label{theo:trace_estim_plus}
Assume that $A$ has a smooth boundary. Then for any $p \in (0,1]$, there exists $C$ such that $d_k^{-1} \op{tr} (T_A^p ( 1-T_A)^p ) \leqslant C k^{-1/2}$.
\end{theo}

\begin{proof} 
Introduce the normalized coherent state $f_x = e_x / \|e_x \|$. Since $\|e_x\|^2 = \Pi(x,x)$, we have for any endomorphism $S$ of $\mathcal{H}_k$
$$ d_k^{-1} \op{tr} S = d_k^{-1} \int_M \langle S e_x , e_x \rangle \; \mu (x) \leqslant C  \int_M \langle S f_x, f_x \rangle \; \mu (x)$$
by using (\ref{eq:diag}). Let us apply this to $S = h(T_A)$ with $h(x) = x^p (1-x)^p$. Since $h(x) \leqslant x^p$ on $[0,1]$, we have $\langle h(T_A) f_x, f_x \rangle \leqslant \langle T_A^p f_x , f_x \rangle $. Writing $f_x$ in an orthonormal eigenbasis of $T_A$ and using that the function $x^p$ is concave for $0<p<1$, we get that $\langle T_A^p f_x, f_x \rangle \leqslant (\langle T_A f_x , f_x \rangle)^p  = \| f_x \|_A^{2p}$. Now if $x\notin A$, we have by Proposition \ref{prop:loc_et_coh} and (\ref{eq:diag}) that$$ \| f_x \|_A^{2p} \leqslant C e^{ - k \rho (x)^2/C}  + \bigo ( k^{-\infty}) $$   
Integrating this inequality over $A^c$ in local coordinates and doing a rescaling by a factor $\sqrt{k}$, we get that
$$ \int_{A^c} \langle h(T_A) f_x, f_x \rangle \leqslant C k^{-1/2}$$
Arguing similarly, we have $h(x) \leqslant (1-x)^p$ on $[0,1]$ and $1-T_A = T_{A^c}$ so that $ \langle h(T_A) f_x, f_x \rangle \leqslant \| f_x \|_{A^c}^{2p}$ and by applying Proposition \ref{prop:loc_et_coh} to $A^c$
$$ \int_{A} \langle h(T_A) f_x, f_x \rangle \leqslant C k^{-1/2}$$
which concludes the proof.
\end{proof}

\begin{cor} Assume that $A$ has a smooth boundary. Then for any $0 < \ep < 1$ and $N \in \N$, we have 
\begin{gather*}
  d_{k}^{-1} \bigl| \op{sp} (T_A) \cap [0, k^{-N}] \bigr| = \nu (A^c) + \bigo(k^{-1/2+\ep}) , \\ 
d_{k}^{-1} \bigl| \op{sp} (T_A) \cap [k^{-N},1-k^{-N}]\bigr| =  \bigo(k^{-1/2+\ep}), \\ 
 d_{k}^{-1} \bigl| \op{sp} (T_A) \cap [1 - k^{-N}, 1] \bigr|  = \nu (A) + \bigo(k^{-1/2+\ep}) .
\end{gather*}
\end{cor}

\begin{proof} Let $\la \in [0,1]$ and $p \in (0,1)$. Define $\al := \la^{-p} ( 1- \la ) ^{1-p}$ and $h_p (x) := x^p (1-x)^p$. Then one proves that $$f(x) := 1-x  - \al h_p(x)$$ satisfies $f(x) \leqslant 1_{[0,\lambda]}$ on $[0,1]$. So $\op{tr} f(T_A) \leqslant  \bigl| \op{sp} (T_A) \cap [0,\lambda] \bigr|$. By Theorem \ref{theo:trace_estim_plus}, 
$$ d_k^{-1}  \op{tr} f(T_A) = \nu (A^c) + \bigo ( k^{-1}) + \al \bigo ( k^{-1/2}).$$
Now choose $\la = k^{-N}$ so that $\al k^{-1/2} \sim k^{-1/2 +pN} = \bigo ( k^{-1/2 + \ep})$ if $pN < \ep$. So working with the convenient $p$, we get that 
$$ d_{k}^{-1} \bigl| \op{sp} (T_A) \cap [0, k^{-N}] \bigr| \geqslant \nu (A^c) + \bigo(k^{-1/2+\ep})$$
By the same argument applied to $A^c$ we have
$$  d_{k}^{-1} \bigl| \op{sp} (T_A) \cap [1 - k^{-N}, 1] \bigr|  \geqslant \nu (A) + \bigo(k^{-1/2+\ep}) $$
which concludes the proof.  
\end{proof}

\begin{rem} \label{rem:nombre_vap_facile}
We also have that for any $\ep \in (0,1/2)$, 
$$ \bigl| \op{sp} ( T_A) \cap [\ep, 1- \ep ] \bigr|  = \bigo ( k^{n-1/2})$$
by essentially the same proof with $p=1$. Of course, we will do much better later by giving an equivalent. 
\end{rem}

\section{Degenerate stationary phase} 

Consider the following integral
$$ I_k = \int_D e^{-k \varphi (t,s)} a(t,s) \; dt \; ds, $$
where $(t,s) \in \R^{n} \times \R^p$, $a: \R^{n+p} \rightarrow \C$ is a smooth compactly supported function and 
\begin{itemize} 
\item $D$ is measurable subset of $\R^{n+p}$ which is conic ($\forall \la >0$, $\la D  = D$) and satisfies $| s | \leqslant C | t|$ for any $(s,t) \in D$, with $C$ independent of $(s,t)$. 
\item $\varphi : \R^{n+p} \rightarrow \C$ is a smooth function such that $\varphi$ and $d \varphi$ vanish along $L = \{ (t,s)\in \R^{n+p} / t=0\}$. The real part $\varphi^r$ of $\varphi$ satisfies $\varphi^r(t,s) >0 $ if $t \neq 0$. Furthermore $( \partial_{t_i} \partial_{t_j} \varphi^r (0,0))_{i,j =1, \ldots n}$ is a positive definite matrix. 
\end{itemize}
Finally $k$ is any positive number. We will describe the asymptotic behaviour of $I_k$ in the large $k$ limit.

\begin{theo} \label{theo:degen-stat-phas}
$I_k$ has the following asymptotic expansion in the limit where $k \rightarrow \infty$: for any $N \in \N$, 
\begin{gather} \label{eq:dev_as}
 I_k = k^{ - (n+p)/2} \Biggl( \sum_{\ell =0 }^{N} b_{\ell} \; k^{-\ell /2} + \bigo \Bigl( k^{ -(N+1)/2}\Bigr) \Biggr)
\end{gather}
where the $b_{\ell}$ are complex numbers, the leading coefficient being given by
\begin{gather} \label{eq:lead_def}
 b_0 = a(0,0) \int_{D} e^{ -q(t) } \; dt \; ds
\end{gather}
where $q$ is the quadratic form 
$q(t) = \tfrac{1}{2} \sum_{i,j} \partial_{t_i} \partial_{t_j} \varphi (0,0) t_i t_j$.
\end{theo}

For $p =0$ and $D = \R^n$, the result is well-known and is sometimes called the Laplace method. In our applications, we will always have $p=1$ and  $D$ will be defined by inequalities $l_1(t,s) \geqslant 0, \ldots , l_m (t,s) \geqslant 0$ where $l_1, \ldots, l_m$ are linear form of $\R^{n+p}$. 

The integral defining the leading order term in (\ref{eq:lead_def}) may be rewritten as follows 
$$  \int_{D} e^{-q(t) } \;dt \;ds = \int_{\R^{n}} e^{ -q(t) } v(t) \; dt $$ 
where $v(t) = \int_{D_t} ds $ is the volume of $D_t = \{ s / (s,t) \in D \}$. 
Since $|s| \leqslant C|t|$ on $D$, $D_t$ is bounded so $v(t)$ is finite. Since $D$ is conic, $v( \la t ) = \la^p v (t)$ for all positive $\la$. The real part of $q$ is positive definite by assumption, so the integral of $ve^{-q}$ converges.

Before we prove the theorem, let us explain how we can compute the other coefficients $b_{\ell}$. First, for any multi-indices $\al \in \N^n$, $\be \in \N^p$, we have
\begin{gather} \label{eq:integral}
 \int_D e^{-k q(t) } t^\al s^{\be} \; dt \; ds = k^{-( |\al| +|\be| + n +p )/2} I(\al, \be) 
\end{gather}
where 
$$I(\al, \be) = \int_D e^{ - q(t)} t^\al s^\be \; dt \; ds = \int_{\R^n} e^{-q(t)} t^\be v_{\al} (t) \; dt$$
with $v_{\al} (t) = \int_{D_t} s^\al \; ds$. Observe that $v_{\al} (\la t) = \la^{p+ |\al|} v_{\al} (t)$, so that these integrals converge.

Second, let $r(s,t) = \varphi (s,t) - q(t)$. We claim that $r$ vanishes to third order at the origin. Indeed, since $\varphi$ and $d \varphi$ vanish along $\{ t = 0 \}$, we have $\partial_{s_i} \partial_{s_j} \varphi (0,0) =0$ and $\partial_{s_i} \partial_{t_j} \varphi (0,0) =0$. So $q$ is the quadratic part of $\varphi$ at the origin. 

Now write $e^{-k \varphi} a= e^{-k q} e^{ -k r } a$. Then replacing in this expression $r$ and $a$ by their Taylor expansions at the origin and $e^{-kr}$ by the series $\sum_{\ell} (-k r)^\ell /\ell!$, we get a series of the form 
$$ e^{-k \varphi(t,s)} a(t,s) \equiv e^{-k q(t) } \sum_{m, \al, \be} a_{m, \al, \be } k^{-m} t^\al s^\be$$
where the coefficients $a_{m, \al, \be}$  are complex numbers. Now we can compute formally the integral $I_k$ by using (\ref{eq:integral})
\begin{xalignat}{2} \notag 
I_k & \equiv   k^{-(n+p)/2} \sum _{m, \al, \be} a_{m, \al, \be } I(\al, \be) k^{-m - (|\al|+ |\be|)/2} \\ \label{eq:exp}
& \equiv   k^{-(n+p)/2} \sum_{\ell \in \Z}  k^{-\ell/2} \sum _{2 m + |\al|+|\be| =\ell } a_{m, \al, \be } I(\al, \be) .
\end{xalignat}
Because $r$ vanishes to third order at the origin, all the $a_{m, \al, \be} $ are zero as soon as $2 m + |\al|+|\be| <0$. So we can restricts the first sum in (\ref{eq:exp}) to $\ell \in \N$. For the same reason, the second sum in (\ref{eq:exp}) is actually finite for any $\ell$. This formal computation gives the correct asymptotic expansion by the following theorem. 

\begin{theo} \label{theo:coefficients}
The coefficient $b_\ell$ in the expansion (\ref{eq:dev_as}) is given by
\begin{gather} \label{eq:coef_b_ell} 
 b_{\ell} = \sum_{2m + |\al | + | \be | = \ell } a_{m, \al, \be} I(\al, \be) .
\end{gather}
\end{theo}

\begin{rem} \label{rem:symetrie}
If we replace $D$ with $-D$, then $I(\al, \be)$ is multiplied by $(-1)^{|\al | + |\be|}$ so $b_{\ell}$ is multiplied by $(-1)^{\ell}$. So if $D$ is symmetric with respect to the origine ($D = -D$), then for any odd $\ell$, $b_{\ell} =0$. \qed
\end{rem}

We now procced to the proof of Theorems \ref{theo:degen-stat-phas} and \ref{theo:coefficients}.
\begin{lemma} \label{lem:support}
If $a$ is zero on a neighborhood of the origin, then $I_k = \bigo( e^{-k/C})$.
\end{lemma}
\begin{proof} 
Since $|s| \leqslant C|t|$ on $D$, we have $D\cap L =\{0 \}$. So if $a =0$ on a neighborhood of $0$, then the support of $a$ is disjoint from $L$. Since the real part of $\varphi$ is $>0$ outside $L$, we obtain  $|e^{-k\varphi} | = \bigo ( e^{-k/C})$ on the support of $a$. 
\end{proof} 

\begin{lemma}  \label{lem:exponential}
For any $N\in \N$, we have
$$ I_k = \int_D e^{-k q(t)} a_N (t,s) \; dt \; ds +  k^{- (n+p)/2} \bigo ( k^{- (N+1)/2}) $$
with $a_N (t,s) = a(t,s) \sum_{\ell=0} ^N ( -k r(t,s))^\ell/ \ell! $.
\end{lemma}

\begin{proof} For any $z \in \C$, we have that 
$$ \Bigl| e^z -1 - z - \ldots - z^N/N! \Bigr| \leqslant \frac{e^{|z|} | z|^{N+1} }{ (N+1)!}.$$  
Applying this to $z = - k r(t,s)$, we obtain that
$$ e^{ - k \varphi (t,s) }  a(t,s) = e^{-k q(t)} a_N (t,s) + R_N(t,s)$$ 
where the remainder satisfies
\begin{gather} \label{eq:estim_reste_1}
 | R_N(t,s) | \leqslant C e^{-k q(t) + k | r(t,s) | } | k r(t,s) |^{N+1} .
\end{gather}
Since $|s| \leqslant C|t|$ on $D$, we have $( |s| + |t|)^2 \leqslant C q(t) $ on $D$. Since $r = \bigo ( (|s|+|t|)^3)$ at the origin, it comes that $|r(t,s)| \leqslant q(t) /2$ on $V \cap D$ where $V$ is a neighborhood of the origin. So we get
\begin{gather} \label{eq:14}
 e^{-k q(t) + k | r(t,s)| } \leqslant e^{-k q(t) /2} \text{ on } V \cap D. 
\end{gather}
By lemma \ref{lem:support}, we may assume that $a$ is supported in $V$, so the same holds for $R_N$. Using that $|r(t,s)| = \bigo ( (|t| +|s|)^3) = \bigo ( |t|^3)$ on $D$, we deduce from (\ref{eq:estim_reste_1}) and (\ref{eq:14}) that 
$$ |R_N(t,s) | \leqslant C  k^{N+1} e^{ -k q(t)/2} |t|^{3(N+1)} \text{ on } V \cap D .$$
Integrating with respect to $s$ on $D_t$ and using that the volume of $D_t$ is $\bigo ( |t|^p)$, we get
\begin{xalignat*}{2}  
\Biggl| \int_D R_N(t,s) \; dt \; ds\Biggr|  & \leqslant C k^{N+1} \int_{\R^n}e^{-kq(t)/2} | t|^{p+ 3 (N+1)} \; dt \\
& = C k^{N+1} k^{-(n+p + 3(N+1))/2} \int_{\R^n} e^{- q(t) } |t |^{p+ 3 (N+1)} \; dt \\
& = \bigo (  k^{ -(n+p +N+1)/2} )
\end{xalignat*}
which concludes the proof
\end{proof}

\begin{lemma} \label{lem:taylor_expansion}
Let $f :\R^{n+p} \rightarrow \C$ be a smooth function.
\begin{enumerate} 
\item \label{item:1}
if $f$ is compactly supported and $|f(t,s)| = \bigo (( |t| + |s| )^N)$ at the origin, then
$$ \int _D e^{-k q(t) } f(t,s) \; dt \; ds = k^{-(n+p)/2} \bigo (k^{-N/2})$$
\item \label{item:2}
if $f =0$ on a neighborhood of the origin and $|f(t,s)|=  \bigo (( |t| + |s| )^N)$ on $\R^{n+p}$, then 
$$    \int _D e^{-k q(t) } f(t,s) \; dt \; ds = \bigo ( e^{-k/C})$$
\item \label{item:3}
if $f$ is compactly supported and $f(t,s) = g (t,s) + \bigo ( (| t| + |s| )^{N})$  at the origin with $g$ polynomial, then 
$$   \int _D e^{-k q(t) } f(t,s) \; dt \; ds = \int _D e^{-k q(t) } g(t,s) \; dt \; ds + k^{-(n+p)/2} \bigo ( k^{-N/2}) .$$
\end{enumerate}
\end{lemma}
\begin{proof} 
For the first assertion, we use that $f(s,t) = \bigo ( |t|^N)$ on $D$ and that the volume of $D_t$ is $\bigo ( |t|^p)$, so 
$$ \Biggl| \int _D e^{-k q(t) } f(t,s) \; dt \; ds \Biggr| \leqslant C \int_{\R^n} e^{-k q(t)} |t|^{N+p} \; dt = \bigo( k^{-(n+p+N)/2}) .$$
For the second assertion, by assumption, $f(t,s) \neq 0$ implies that $|t| + |s| \geqslant \epsilon >0$. If furthermore $(t,s) \in D$, then $| s| \leqslant C|t|$ so $|t| +|s| \leqslant ( 1+C) |t|$ so $|t| \geqslant \epsilon/(1+C) = \epsilon '$.
Using that $|f(t,s)|  = \bigo (( |t| + |s| )^N)$ on $\R^{n+p}$, we get
$$  \Biggl|  \int _D e^{-k q(t) } f(t,s) \; dt \; ds \Biggr| \leqslant C \int_{\{ |t| \geqslant \epsilon'\}} e^{-k q( t) }  | t|^{N+p} dt = \bigo ( e^{-k/C}).$$ 
which proves the second assertion.  Let $\rho : \R^{n+p} \rightarrow \C$ be smooth compactly supported and equal to $1$ on a neighborhood of the origin. Applying the first assertion to  $\rho ( f-g)$ and the second assertion to $(1-\rho ) (f-g)$ we deduce the third assertion.  
\end{proof}
To deduce Theorems \ref{theo:degen-stat-phas} and \ref{theo:coefficients}, it suffices now to start with Lemma \ref{lem:exponential} and to replace $a_N$ by its Taylor expansion, the integral of the remainder being controlled by the third assertion of Lemma \ref{lem:taylor_expansion}. 

\section{The trace of $T_A^p- T_A^{p+1}$}
Let us denote by $\Pi(x,y)$ the reproducing kernel of $L$. For any $p \in \N^*$, introduce for any $x_1, \ldots , x_p \in M$ the quantity
$$\Pi_p (x_1, \ldots , x_p ) := \Pi (x_1,x_2) \Pi(x_2,x_3) \ldots \Pi(x_{p-1}, x_p) \Pi(x_p,x_1).$$
It belongs to $L_{x_1} \otimes \con{L}_{x_2}\otimes L_{x_2} \otimes \ldots \otimes \con{L}_{x_p} \otimes L_{x_p} \otimes \con{L}_{x_1}$. Using the metric of $L$, we can identify naturally each factor $L_{x_i} \otimes \con{L}_{x_i}$ with $\C$.  So we will view $\Pi_p ( x_1, \ldots, x_p)$ as a complex number and $\Pi_p$ itself as a complex valued function of $M^p$. 

$\Pi_1 (x) = \Pi(x,x)$ is the restriction to the diagonal of the reproducing kernel. $\Pi_2 (x,y) = \Pi ( x,y) \Pi(y,x) = |\Pi (x,y) |^2$ is the square norm of the reproducing kernel. The subsequent terms do not have such an easy interpretation, for instance  $\Pi_3 (x,y,z) = \Pi(x,y) \Pi(y,z)\Pi(z,x)$. Observe the symmetry $\Pi_p ( x_1, \ldots, x_p) = \Pi_p ( x_2, x_3, \ldots , x_p, x_1)$. Furthermore, by the reproducing property
\begin{gather} \label{eq:15}
 \int_M \Pi_{p+1} (x,x_1, \ldots , x_{p}) \; \mu (x) = \Pi_p ( x_1, \ldots , x_p).
\end{gather}
We will use these kernels to compute the trace of $T_A^p$. 
\begin{lemma} \label{lem:les_traces}
For any $p\in \N^*$, we have 
\begin{gather} \label{eq:trace_T_Ap}
 \op{tr} (T_A^p) = \int_{A^p} \Pi_p  ( x_1, \ldots, x_p ) \; \mu(x_1) \ldots \mu (x_p) \\
\label{eq:la_trace}
 \op{tr} ( T_A^p - T_A^{p+1} ) = \int_{A^p \times A^c} \Pi_{p+1} (  x_1, \ldots, x_{p+1} ) \; \mu(x_1) \ldots \mu (x_{p+1}).
\end{gather}
\end{lemma}

\begin{proof} 
Recall that the Schwartz kernel of an endomorphism $P$ of $\mathcal{H}$ is the the holomorphic section $K$ of $L \boxtimes \con{L}$ such that $(Ps )(x) = \int_M K(x,y) . s(y)  d\mu (y)$. Here the dot stands for the contraction $\con{L}_y \otimes L_y =\C$ induced by the metric.  Since $(T_A s )(x) = \int_A K(x,y).s(y) \mu (y)$, the Schwartz kernel of $T_A^p$ is 
$$ (x,y) \rightarrow \int_{A^p} \Pi(x,x_1).\Pi(x_1,x_2) \ldots \Pi(x_{p-1}, x_p) . \Pi (x_p, y) \;  \mu(x_1) \ldots \mu (x_p) $$
The trace of an endomorphism of $\mathcal{H}$ being equal to the integral of its Schwartz kernel on the diagonal, we obtain
\begin{xalignat}{2} \label{eq:1}
 \op{tr} (T_A^p) & = \int_{M\times A^p } \Pi_{p+1}  (x, x_1, \ldots, x_p ) \; \mu (x) \mu(x_1) \ldots \mu (x_p) \\ \label{eq:2}
& = \int_{A^p } \Pi_p  ( x_1, \ldots, x_p ) \;  \mu(x_1) \ldots \mu (x_p) 
\end{xalignat}
by (\ref{eq:15}). 
This proves (\ref{eq:trace_T_Ap}). Computing the trace of $T_A^p$ with (\ref{eq:1}), the trace of $T_A^{p+1}$ with (\ref{eq:2}) and using the symmetry of $\Pi_{p+1}$ we obtain
\begin{xalignat*} {2} 
  \op{tr} ( T_A^p - T_A^{p+1} ) & = \int_{A^p \times M} \Pi_{p+1} (  x_1, \ldots, x_{p+1} ) \; \mu(x_1) \ldots \mu (x_{p+1}) \\ & - \int_{A^{p+1}} \Pi_{p+1}  ( x_1, \ldots, x_{p+1} ) \; \mu(x_1) \ldots \mu (x_{p+1}) 
\end{xalignat*}
which implies (\ref{eq:la_trace}) because $(A^p \times M )\setminus A^{p+1} = A^p \times A^c$.
\end{proof}

The previous considerations may be applied to $\Hilb_k$. So for each $p$, we have a family of kernels $K_{p,k}$. We are going to describe their asymptotic behaviour as $k$ tends to $\infty$. Let us introduce some notations. Let $\Delta_p$ be the diagonal map $M \rightarrow M^p$ sending $x$ into $(x,\ldots ,x)$. To describe the Hessian $H$ of a function $f : M^p \rightarrow \C$ at a critical point $\Delta_p (x) $, we will use the following block decomposition. For $i=0, \ldots , p$, let $u_i$ be the linear map from $T_xM$ to $T_{\Delta_p(x)} M^p \simeq (T_xM)^p $ given by 
$$ u_0 (X)= (X,\ldots ,X), \qquad u_i(X) = (0,\ldots,0, X, 0, \ldots ,0) \text{ if } i \neq 0$$
with the $X$ in the $i$-th position. Then we define the bilinear forms $H_{ij}$ of $T_xM$ by
$$ H_{ij} (X,Y) = H(u_i(X), u_j(Y)), \qquad i, j =0, \ldots, p$$
From the description of the Bergman kernel given in Theorem \ref{theo:bergman-kernel}, we can deduce the following result. 
\begin{theo} \label{theo:Kp_asymptotic}
$\Pi_{k,p}$ has the following form 
$$ \Pi_{p,k} = \Bigl( \frac{k}{2 \pi} \Bigr)^{pn} e^{-k\varphi} f(\cdot, k) + r_k$$
where $\varphi$, $f (\cdot, k)$ and $r_k$ are functions in $\Ci ( M^p, \C)$ such that 
\begin{enumerate} 
\item $\varphi$ and $d\varphi$ vanish along $\Delta_p (M)$. The real part of $\varphi$ is $>0$ outside $\Delta_p (M)$, its Hessian being non degenerate in the direction transversal to $\De_p (M)$. For any $x \in M$, the Hessian $H$ of $\varphi$ at $\Delta_p(x)$ is given by 
\begin{gather} \label{eq:Hessian_phi} 
\begin{split} 
 H_{ii} =   g_x \qquad    \forall i = 1, \ldots ,p \\
H_{i, i+1}  =  \tfrac{i}{2} \om_x - \tfrac{1}{2} g_x  \qquad  \forall i = 1, \ldots ,p-1
\end{split}
\end{gather}
and $H_{ij} = 0$ if $i=0$ or $j=0$ or $|i-j|\geqslant 2$. Here $g_x$ and $\om_x$ are the Riemannian metric and the symplectic form on $T_xM$.   
\item the sequence $f(\cdot,k)$ has an asymptotic expansion 
$$f(\cdot, k) = f_0+ k^{-1} f_1 + \ldots + k^{-N} f_N + \bigo ( k^{-N+1}), \qquad \forall N$$
where the $\bigo$ is uniform on $M^p$ and the coefficients $f_\ell$ are in $\Ci ( M , \C)$. Furthermore, the restriction of $f_0$ to $\Delta_p(M)$ is constant equal to $1$. 
\item the remainder $r_k$ is $\bigo(k^{-N})$ uniformly on $M^p$ for any $N$. 
\end{enumerate}
\end{theo}

\begin{proof}  Using the notation introduced in Theorem  \ref{theo:bergman-kernel}, the section $E$ is equal to 1 on the diagonal so there exists a function $\varphi$ vanishing along $\Delta_p (M)$ such that 
$$ e^{-\varphi (x_1, \ldots , x_p) } = E(x_1,x_2) E(x_2,x_3) \ldots E(x_p,x_1)$$
on a neighborhood of the diagonal. Similarly, we set 
$$ f(x_1,\ldots, x_p,k) = a(x_1,x_2,k) a(x_2,x_3,k) \ldots a(x_p, x_1,k).$$
Then the properties of $\varphi$ and $f(\cdot, k)$ follows from the ones of $E$ and $a(\cdot,k)$. In particular, the real part of $\varphi$ is 
$$ \op{Re} \varphi (x_1, \ldots , x_p) = \tfrac{1}{2} \bigl( \varphi_E (x_1,x_2) + \varphi_E(x_2, x_3) + \ldots + \varphi_E( x_p, x_1) \bigr)$$
where $\varphi_E = - 2 \ln |E|$. Recall that the Hessian of $\varphi_E$ is non negative, its radical being the tangent space of the diagonal. Using that a sum  of non negative bilinear forms is non negative, and that the radical of the sum is the intersection of the radicals of the summands, we deduce that the Hessian of $\op{Re} \varphi$ is non negative, its radical being the tangent space to $\De_p (M)$. 

Let us compute the Hessian $H$ of $\varphi$ at $\De_p (x)$.  First for any vector fields $Y_1,\ldots ,Y_p$ of $M$, we have
\begin{xalignat}{2} \notag
 - \bigl( \mathcal{L}_{(Y_1,\cdots, Y_p)} \varphi\bigr) (x_1,\ldots, x_p) & = \tfrac{1}{i} \Bigl( \al_E (Y_1,Y_2) (x_1,x_2) + \al_E(Y_2,Y_3) (x_2,x_3) \\ \label{eq:der_phi}
& + \ldots +\al_E( Y_p, Y_1) (x_p, x_1) \Bigr) 
\end{xalignat}
where $\al_E$ has been defined in Theorem \ref{theo:bergman-kernel}. 
Since $\al_E$ vanishes along the diagonal, $d \varphi$ vanishes along $\Delta_p(M)$. Since $\varphi$ and $d \varphi$ are zero along $\Delta_p(M)$, we have that $H_{ij} = 0 $ if $i =0$ or $j=0$. It is also clear from (\ref{eq:der_phi}) that $H_{ij} =0$ if $|i-j| \geqslant 2$. Let us compute $H_{11}$
\begin{xalignat*}{3}
 H_{11} ( X,Y)  & = \bigl(\mathcal{L}_{(X, 0, \ldots , 0)} \mathcal{L}_{(Y, 0, \ldots, 0)} \varphi\bigr)( \Delta_p(x) ) && \text{ by definition}\\ 
& = i ( \mathcal{L}_{(X,0) } \al_E(Y,0) + \mathcal{L}_{(0,X) } \al_E(0,Y)  \bigr) (x,x)  && \text{ by } (\ref{eq:der_phi})\\
& = i \bigl( \om( X^{0,1}, Y) - \om ( X^{1,0}, Y) \bigr) (x) \qquad && \text{ by } (\ref{eq:der_seconde_E}) \\
& = - \om (jX,Y)(x) \qquad \text{ because }   X^{0,1}-  X^{1,0}= ijX && \\
& = g(X,Y)(x) 
\end{xalignat*}
 Let us compute $H_{12}$
\begin{xalignat*}{2}
 H_{12} ( X,Y)  & = \bigl(\mathcal{L}_{(X, 0, \ldots , 0)} \mathcal{L}_{(0,Y, 0, \ldots, 0)} \varphi\bigr)( \Delta_p(x) ) && \text{by definition} \\ 
& = i ( \mathcal{L}_{(X,0) } \al_E(0,Y) \bigr) (x,x)  && \text{by } (\ref{eq:der_phi})\\
& = i \om ( X^{1,0}, Y)  (x) \qquad && \text{by } (\ref{eq:der_seconde_E}) \\
& = \tfrac{i}{2} \om (X,Y)(x) - \tfrac{1}{2} \om ( jX,Y)(x)  && \text{because } X^{1,0}= \tfrac{1}{2} ( X -ijX) \\
& = \tfrac{i}{2} \om (X,Y)(x) - \tfrac{1}{2} g(X,Y)(x) 
\end{xalignat*}
The computation of $H_{ij}$ with $j =i$ or $i+1$ is similar. 
\end{proof}

\begin{theo} \label{theo:asymp_trace}
For any $p \in \N^*$ and any domain $A$ with a smooth boundary, we have the asymptotic expansion
$$ \op{tr} (T_A^{p} - T_A^{p+1}) = k^{n-1/2} \Biggl( \sum_{\ell =0 }^{N} b_\ell (A) k^{-\ell/2} + \bigo ( k^{-(N+1)/2}) \Biggr) , \qquad \forall N>0 $$
where the coefficients $b_\ell(A)$ are real numbers. Furthermore, for any $\ell$ 
\begin{gather} \label{eq:sym_coef}
 b_\ell (B) = (-1)^{\ell} b_{\ell }(A). 
\end{gather}
if $B =M \setminus \op{int}(A)$.
\end{theo}

\begin{proof} 
By Lemma \ref{lem:les_traces}, $\op{tr} (T_A^{p} - T_A^{p+1})$ is the integral of $K_{p+1}$ on $A^p \times B$. By Theorem \ref{theo:Kp_asymptotic}, it suffices to consider the integral 
\begin{gather} \label{eq:Ika}
 I_k(a) = \int_{A^p \times B} e^{-k \varphi(x_1, \ldots , x_{p+1})} a(x_1, \ldots , x_{p+1}) \; \mu (x_1) \ldots \mu ( x_{p+1}) .
\end{gather}
Here we denote by $\varphi$ the function corresponding to $K_{p+1}$ and not $K_p$. Since the real part of $\varphi$ is $>0$ outside $\Delta_{p+1} (M)$, $I_k(a) = \bigo(k^{-\infty})$ when the support of $a$ does not meet $\Delta_{p+1} (M)$. Observe that $\Delta_{p+1} (M) \cap (A^p \times B) = \Delta_{p+1} ( \partial A)$. So it suffices to do the computation with $a$ supported in an arbitrary small neighborhood  of a point $\De_{p+1} (z)$ where $z \in \partial A$. Choosing coordinates $(x^i)$, we may identify a neighborhood of $z$ with an open ball $U$ of $\R^{2n} $ centered at the origin such that $ U \cap A = \{ x \in U,\;   x^1 \geqslant 0 \}$. So we use $x'=(x^2, \ldots , x^{2n})$ as coordinates on $\partial A \cap U$. We assume that $a$ is compactly supported in $U^{p+1}$, and in this way we can consider that $U = \R^{2n}$. Instead of $(x_1, \ldots , x_{p+1}) \in U^{p+1}$, we will work with $(t_1, \ldots , t_p , \tau) \in U^{p+1}$ defined by 
\begin{gather} \label{eq:newcoord}
 t_1 = x_1 -x_{p+1} , \; t_2 = x_2 - x_{p+1} , \; \ldots, \; t_p = x_p - x_{p+1}, \; \tau = x_{p+1} 
\end{gather}
so that 
$$ x_1 = \tau +t_1  , \; x_2 = \tau +t_2 , \; \ldots, \; x_p = \tau + t_p, \; x_{p+1} = \tau .$$   
With these new coordinates, we have 
\begin{xalignat}{2}  \label{eq:3}
\Delta_{p+1} (U) & = \{ t_1 = \ldots = t_p =0\} \\ \Delta_{p+1} ( \partial A) &= \{ t_1 = \ldots = t_p = 0, \; \tau^1=0\}
\end{xalignat}
so we can use $\tau'=(\tau^2, \ldots , \tau^{2n})$ as coordinates on  $\Delta_{p+1} ( \partial A)$. Furthermore 
$$ A^p \times B = \bigl\{ t_1^1, t_2^1, \ldots , t_p^1 \geqslant 0 \text{ and }  0 \leqslant -\tau^1 \leqslant \min ( t_i^1) \bigr\} $$
Now we first consider that $\tau'$ is fixed and integrate $e^{-k \varphi} a$ with respect to $t = (t_1, \ldots , t_p)$ and $s= - \tau^1$. This integral has an asymptotic expansion given by Theorem \ref{theo:degen-stat-phas}. The assumptions are easily checked. Indeed the real part of $\varphi$ has a non degenerate Hessian in the direction corresponding to $t_1, \ldots, t_p$ because $\Delta_{p+1} (U)$ is given by (\ref{eq:3}). The coordinates $(t,s)$ introduced here correspond to the ones of Theorem  \ref{theo:degen-stat-phas}, the domain $D$ being given by 
\begin{gather} \label{eq:domain_D}
D=\{ (t,s) \in \R^{2np} \times \R, \; 0 \leqslant s \leqslant \min ( t_1^1, \ldots , t_p^1) \} .
\end{gather}
Then we integrate with respect to $\tau'$. This proves that $\op{tr} (T_A^{p} - T_A^{p+1})$ has an asymptotic expansion of the form $k^d \sum b_\ell (A) k^{-\ell/2}$. The power of $k$ is 
$$d = n (p+1) - \dim D =n-1/2. $$  It remains to explain the relation (\ref{eq:sym_coef}). Actually, the computation with $B$ is exactly the same except that we integrate on $B^p \times A$ instead of $A^p \times B$, which amounts to replace $D$ with $-D$. So the relation (\ref{eq:sym_coef}) follows from Remark \ref{rem:symetrie}. 
\end{proof}

\begin{theo} \label{theo:leading_order_coefficient}
The leading order coefficient in the asymptotic expansion of Theorem \ref{theo:asymp_trace} is given by $b_0 ( A) = C_{p,n} \op{vol}(\partial A)$ where $C_{p,n}$ is a universal positive constant depending only on $p$ and the complex dimension $n$ of $M$, and $\op{vol} (\partial A) $ is the Riemannian volume of $\partial A$. 
\end{theo}

In the proof we will obtain the following formula for  $C_{p,n}$
\begin{gather} \label{eq:integrales}
C_{p,n} = (2 \pi)^{-n(p+1)} \int_{D} e^{-q(t) } \; dt \; ds 
\end{gather}
where $D$ is the domain of $\R^{2np} \times \R \ni (t,s)$ defined by (\ref{eq:domain_D}) and $q$ is the quadratic form 
\begin{gather} \label{eq:quadratique_modele}
q(t) = \frac{1}{2} \sum_{i, j  = 1, \ldots , p} \langle M_{ij} t_i, t_j \rangle, \qquad t = (t_1, \ldots , t_p) \in (\R^{2n})^p .
\end{gather}
Here the $M_{ij}$ are square  matrices of size $2n$ defined by $M_{ii} = \op{id}_{2n}$, 
\begin{gather} \label{eq:special_matrice}
 M_{i, i+1} = \tfrac{1}{2} \left[ {\begin{array}{cc} - 1 & i \\ -i & -1 
  \end{array} } \right] \otimes \op{id}_n, \qquad M_{i+1,i} = (M_{i,i+1})^{\op{t}} 
\end{gather}
if $i =1, \ldots , p-1$ and $M_{ij}= 0$ when $|i-j|\geqslant 2$. 

 Even if the constants $C_{p,n}$ are perfectly defined with (\ref{eq:integrales}), we won't use this formula to compute them. Actually since by Theorem \ref{theo:leading_order_coefficient}, the constants $C_{p,n}$ are universal, we can compute them on any example. 

\begin{proof} 
Consider again the integral $I_k(a)$ introduced in (\ref{eq:Ika}). At the end of the proof of Theorem \ref{theo:asymp_trace}, we applied Theorem \ref{theo:degen-stat-phas} and conclude that $I_k(a)$ has an asymptotic expansion $k^{-(np+1/2)}( c_0(a)+ k^{-1/2} c_1(a) + \ldots)$. By the same theorem, we can compute the leading order term $c_0(a)$. Using the coordinates $(t_1, \ldots , t_p , \tau)$ introduced in (\ref{eq:newcoord}), we have
$$ c_0 (a) = \int_{\R^{2n-1} }  a ( 0, \tau') m (0,\tau ') I(\tau') \; d\tau'$$
where $m$ is the function defined by
\begin{gather} \label{eq:vol_form_p}
\mu(x_1) \ldots  \mu(x_{p+1}) = m(t_1, \ldots ,t_p ,\tau) dt_1 \ldots dt_p \; d\tau
\end{gather}
and $I(\tau')$ is the integral 
$$ I(\tau') = \int_D e^{-q(t,\tau') } dt \; ds $$
with $t \rightarrow q( t, \tau')$  the Hessian at the origin of $t \rightarrow \varphi ( t, 0 , \tau')$. Our goal is to prove that 
$$ c_0 (a) = C_{n,p} \int_{\partial A } \tilde {a} \; \nu $$  
where $C_{n,p} $ is the integral (\ref{eq:integrales}), $\tilde{a}$ is the restriction of $a$ to $\Delta_{p+1} (\partial A) \simeq \partial A$ and $\nu$ is the Riemannian volume element of $\partial A$. 

Let us temporarily assume that we can choose our coordinates $x^i$ of $M$ so that $(\partial_{x^i})$ is an orthosymplectic frame of $TM$ along $\partial A$, that is for any $z \in \partial A$, $(\partial_{x^i}|_z)$ is an orthornormal basis of $T_zM$ and $j ( \partial_{x^i} ) = \partial_{x^{i+1}}$ at $z$ for $i =1, \ldots, n$. Then the matrices of the bilinear forms $g_z$ and $\frac{i}{2} \om_z - \frac{1}{2} g_z$ are $\op{id}_{2n} $ and $M_{i,i+1}$ given in (\ref{eq:special_matrice}). So computing the Hessian of $t \rightarrow \varphi (t, 0, \tau')$ with the relations (\ref{eq:Hessian_phi}) of Theorem \ref{theo:Kp_asymptotic}, we get that $q(t, \tau' ) = q(t)$ for any $\tau'$ with $q$ the quadratic form introduced in (\ref{eq:quadratique_modele}). Furthermore, the Liouville volume form is $\mu (x) = dx $ and the Riemannian volume element is $\nu = d \tau'$ in this case. And the proof is complete. 

Of course, our assumption was clearly unrealistic. Actually, we can assume that $|dx^1| =1$ on $\partial A$. To achieve this, it suffices to multiply $x^1$ be the convenient function $f(x^2, \ldots, x^n)$. Also, we can choose a local orthosymplectic frame of $TM$. 
Using this we will prove that 
$$ m(0, \tau') I ( \tau') d \tau ' = C_{p,n} \nu. $$ 
Here on the left hand side we view $\tau'$ as coordinates on $\partial A$ and on right hand side $\nu $ is the Riemannian volume element of $\partial A$. 

Let us describe $\la = m(0, \tau') I ( \tau') d \tau '$ in terms of linear data. Let $z \in \partial A$, $E = T_zM$, $\mu_{p+1} \in \wedge^{\op{top}}(E^{p+1})^*$ be the value of the volume form (\ref{eq:vol_form_p}) at $\Delta_{p+1} (z)$ and $H$ be the Hessian of $\varphi$ at $\Delta_{p+1} (z)$.   Introduce the maps 
\begin{gather*} 
 \Delta : E \rightarrow E^{p+1}, \quad v \rightarrow (v,\ldots, v), \\ 
i : E^p \times \R \rightarrow E^{p+1}, \quad (x,r) \rightarrow (x,0) - r\Delta (u)
\end{gather*}
where $u \in T_z M$ is the outgoing normal vector. Then for any $X \in \wedge ^{\op{top}} E$, we have
\begin{gather} \label{eq:lambda} 
 \la_z ( X) =  \int_{D_E} e^{- q_E  } i^* ( \iota_{\Delta(X)}  \mu_{p+1}) 
\end{gather}
where $q_E: E^p \times \R \rightarrow \C$ is given by $q_E ( y) = \frac{1}{2}  H( i(y) , i(y))$ and $D_E = i^{-1} (F^p \times (-F))$ with $F = \{ x \in E, g(x,u) \geqslant 0 \}$.  So $D_E$ consists of the $(x_1, \ldots , x_p, r) \in E^p \times \R$ such that $ 0 \leqslant r \leqslant \min_i  g(x_i, u )$

To prove (\ref{eq:lambda}), it suffices to compute the integral in terms of linear coordinates of $E$ dual to the basis $(\partial_{x^i})$ and we recover exactly $ m(0, \tau') I ( \tau') d \tau '$. Now if we compute $\la$ with (\ref{eq:lambda}) in an orthosymplectic basis of $E$, we obtain that $\la = C_{p,n} \nu$. 
\end{proof}

We are now going to compute the constants $C_{p,n}$ without using the formula (\ref{eq:integrales}).

\begin{lemma} \label{lem:coef_n}
We have $C_{p,n+1} = C_{p,1} ( 2\pi) ^{-n}$.
\end{lemma}

\begin{proof} 
Consider two positive line bundles $L_1 \rightarrow M^1$ and $L \rightarrow M$ where $M^1$ and $M$ have complex dimension $1$ and $n$ respectively. Then $L' =L_1 \boxtimes L$ is a positive line bundle on $M^1 \times M$ and with the obvious notation $\mathcal{H}_k'= \mathcal{H}^1_k \otimes \mathcal{H}_k$. Choose a domain $A \subset M^1$ with a smooth boundary. Then $T_{A \times M } = T_{A} \otimes \op{id}_{\mathcal{H}_k}$. So for any function $f: [0,1] \rightarrow \C$, 
$$ \op{tr}(f(T_{A \times M})) = \bigl( \dim \mathcal{H}_k \bigr) \op{tr} f(T_A) $$
By Theorem \ref{theo:leading_order_coefficient} and the dimension estimate (\ref{eq:dim_est}), we get
$$ k^{n+1/2} C_{p, n+1} \op{vol}(( \partial A) \times M) = \Bigl( \frac{k}{2\pi} \Bigr)^n \op{vol} (M) \; k^{1/2}\; C_{p,1}  \op{vol} ( \partial A) $$
which gives the result.
\end{proof}

Denote by $\er$ the function
$$ \er (x) = \frac{1}{\sqrt \pi}\int_{- \infty} ^x e^{-t^2} \; dt , \qquad x \in \R$$
So $\er$ increases from $0$ to $1$ with asymptotic behaviour
\begin{gather} \label{eq:error_Asympt} 
 \er(-x) = 1 - \er (x) =   \frac{e^{-x^2}}{2 \sqrt \pi x} \bigl( 1 + \bigo ( x^{-2}) \bigr) \qquad \text{ as } x \rightarrow \infty .
\end{gather}

\begin{lemma} $C_{p,1} = (2 \pi)^{-1} \int_{\R} g_p(\er (x) ) \; dx$ with $g_p(x) = x^{p} - x^{p+1} $.
\end{lemma}

\begin{proof} Consider the trivial holomorphic line bundle $L_{\C}$ over $\C$ with canonical frame $\si$. Write $z = 2^{-1/2} ( x+ i y)$ and introduce the metric on $L_\C$ such that $|\si (z)|^2 = \exp (-x^2)$. The curvature of the Chern connection is $dz \wedge d \con{z}$, so $L_{\C}$ is positive. The action of $\Z$ on $L_{\C}$ given by $ n. \si (z) = \si (z+in)$ preserves the holomorphic and Hermitian structures. So taking the quotient, we obtain a positive line bundle $L$ on $M = \C / i \Z$. For any $k$, let $\mathcal{H}_k$ be the space of holomorphic sections $s$ of $L^k$ such that $\int_M |s|^2 \; dx dy $ is finite. $\mathcal{H}_k$ is a Hilbert space with scalar product defined by (\ref{eq:scalar_product}). An orthonormal basis of $\Hilb_k$ is given by 
$$ s_{\ell}  ( z) =  (2\pi)^{-1/2} \pi^{-1/4} k^{1/4} e^{ -\ell^2/ 2 k - \ell z} \si( z) ^k , \qquad \ell \in \Z.$$ 
Here we identify sections of $L^k$ with their lifts to $\C$. 
Consider the domain $A = \{ [z] / \; \op{re} z \leqslant 0 \}$ of $M$. We compute that
$$ \langle 1_A s_{\ell} , s_{m} \rangle = \delta_{\ell m}  \er (  \ell/ \sqrt k). $$
So $(s_{\ell})$ is an eigenbasis of the operator $T_A$, which has a discrete spectrum with simple eigenvalues $\er ( \ell/ \sqrt k)$, $\ell \in \Z$. Since the error function satisfies (\ref{eq:error_Asympt}), $g_p ( T_A)$ is a trace class operator and 
\begin{gather} \label{eq:5}
 \op{tr} (g_p ( T_A)) = \sum_{\ell \in \Z} g_p ( \er (\ell / k)) =  \sqrt{k} \int_{\R} g_p ( \er (x) ) \; dx + \bigo ( k^{-\infty})   
\end{gather}
by Lemma \ref{lem:Euler_Maclaurin}. Now, even if $M$ is not compact, we can still compute $\op{tr} g_p (T_A)$ by the method of Theorem \ref{theo:asymp_trace} because the Bergman kernel has the expected asymptotic behavior and is exponentially decreasing at infinity. So we have 
\begin{gather} \label{eq:6}
 \op{tr} (g_p ( T_A)) = \sqrt{k} \; C_{p,1} (2 \pi ) \bigl( 1 + \bigo ( k^{-1/2}) \bigr) 
\end{gather}
because the length of $\partial A$ is $2\pi$. Comparing (\ref{eq:5}) and (\ref{eq:6}), we get the result.
\end{proof}

\begin{lemma} \label{lem:Euler_Maclaurin}
Let $f\in \Ci ( \R, \C)$ be such that for any $\ell \in \N$, its $\ell$-th derivative satisfies $\lim _{x \rightarrow \pm \infty} f^{(\ell)}(x) = 0$ and $\int_{\R} f^{(\ell)}(x) \; dx < \infty$. Then we have in the limit $\tau \rightarrow \infty$
$$ \tau^{-1} \sum_{\ell \in \Z}  f(\ell /\tau ) = \int_{\R} f(x) \; dx + \bigo ( \tau^{-\infty}).$$
\end{lemma}

\begin{proof} 
By Euler-Maclaurin formula \cite[Theorem II.10.2]{HaWa}, we have for any smooth function $g$ and $N \in \N^*$ 
$$ \sum_{i = -n+1}^{n} g(i) = \int_{-n}^{n} g(x) \; dx + \sum_{\ell = 0 }^{N-1} C_{\ell} \bigl( g^{(\ell)} (n) - g^{(\ell) } (-n) \bigr) + \int _{-n}^n a_N(x) g^{(N)}(x) \; dx $$ 
with some constant $C_\ell$ independent of $g$ and $1$-periodic functions $a_N$ independent of $g$. Applying this to $g(x) = f(x/\tau)$ it comes that 
\begin{xalignat*}{2}  
\tau^{-1} \sum_{i = -n+1}^{n} f(i/\tau) & = \int_{-n/\tau}^{n/\tau} f(x) \; dx + \sum_{\ell = 0 }^{N-1} C_{\ell} \tau^{-\ell-1} \bigl( f^{(\ell)} (n) - f^{(\ell) } (-n) \bigr) \\& + \tau^{-N}\int _{-n/\tau}^{n/\tau} a_N(\tau x) f^{(N)}(x) \; dx 
\end{xalignat*}
Taking the limit $n \rightarrow \infty$ and using that the $a_N$ are bounded, we get the result. 
\end{proof}

\section{Asymptotics of $\op{tr}(f(T_A))$} 

For any function $f: [0,1] \rightarrow \C$, let us introduce the quantities
$$ S_k (f) = k^{-n+1/2} \op{tr} ( f(T_A) ) , \qquad I (f) = \int_{-\infty}^{\infty} f(\er (x)) \; dx $$

\begin{theo} \label{theo:pol_convergence}
For any polynomial function $f: [0,1] \rightarrow \C$ such that $f(0) = f(1) = 0$, we have the asymptotic expansion
$$ S_k (f) = \sum_{\ell = 0 }^{N} c_{\ell} (f) k^{-\ell/2} + \bigo ( k^{-(N+1)/2}) , \qquad \forall N\in \N$$
where 
\begin{enumerate}
\item the leading order coefficient is $c_0( f) = (2 \pi)^{-n} \op{vol} ( \partial A) I(f) $. 
\item For any $\ell$,  $c_{\ell} (g) = (-1)^{\ell} c_{\ell} (f) $ where $g(x) = f(1-x)$.
\end{enumerate}
\end{theo}

 So $f(x) = f(1-x)$ implies that $c_\ell (f) = 0$ for odd $\ell$,  $f(x) = - f(1-x)$ implies that  $c_{\ell} (f) =0$ for even $\ell$. 

\begin{proof} 
For $f (x) = h_p (x) = x^{p} - x^{p+1}$, the result has been proved in the previous section. This implies the result for a general $f$ because the $h_p$ form a basis of  the polynomials vanishing at $0$ and $1$. For the last part, observe that $g(T_A) = f(T_{A^c})$. 
\end{proof}

Doing the change of variable $t = \er (x)$, the integral $I(f)$  can be written as
$$ I (f) = \int_0^1 f(t) \; \delta (t) \; dt, \qquad \text{ with } \delta ( t) =\sqrt{\pi} \exp ( (\er^{-1} (t))^2)$$
From the symmetry $\er(x) + \er (-x) =1$, one deduces that $\delta (t) = \delta ( 1-t)$. Furthermore, (\ref{eq:error_Asympt}) implies that 
$$ \delta (t) \sim  \frac{1}{2t} \Bigl( \ln ( 1/t) \Bigr)^{-1/2} \qquad \text{ as } t \rightarrow 0$$
Consequently, the integral $I(f)$ converges for any continuous function $f: [0,1] \rightarrow \C$ which is H\"older continuous at $0$ and $1$ with  $f(0) = f(1)=0$.

\begin{theo} \label{theo:Weyl_law_singular}
For any function $f :[0,1] \rightarrow \C$ continuous and such that $|f(t)| =\bigo ( t^p)$ and $|f(1-t)| = \bigo ( t^p)$ as $t \rightarrow 0$ with $p>0$, we have 
$$ \lim_{k \rightarrow \infty} S_k (f) = (2 \pi)^{-n} \op{vol} ( \partial A) I(f) .$$
\end{theo}

\begin{proof} 
We will mainly use that $S_k$ and $I$ are positive functional in the sense that $f \leqslant g$ implies $S_k (f) \leqslant S_k (g)$ and $I(f) \leqslant I(g)$. Introduce the notation $I_A(f) = (2 \pi)^{-n} \op{vol} ( \partial A) I(f)$. We will now prove the result for particular classes of functions $f$. 
 
1. If $f$ is a polynomial, the result follows from Theorem \ref{theo:pol_convergence}. 

2.  Let $f =gh $ with $g \in \mathcal{C} ([0,1])$ and $h(x) = x(1-x)$. Let $\epsilon >0$. By Weierstrass theorem, there exists a polynomial function $Q:[0,1] \rightarrow \C$ such that $| g - Q | \leqslant \ep$. So $g \leqslant Q + \ep$ which implies $S_k (gh) \leqslant S_k(Qh) + \ep S_k(h)$. By part 1, $S_k (h) \leqslant C$ and $S_k (Qh) \rightarrow I_A(Qh)$. So $\limsup S_k(f) \leqslant I_A(Qh) +\epsilon C$. Furthermore, $Q \leqslant g + \ep$ implies $I_A(Qh) \leqslant I_A( gh) + \ep I_A(h)$. So 
$$ \limsup S_k ( f) \leqslant I_A(f) + \ep C + \ep I_A(h)$$
Taking the limit $\epsilon \rightarrow 0$, we get $\limsup S_k (f) \leqslant I_A(f)$. By exactly the same argument, we also have $\liminf S_k(f) \geqslant I_A (f)$. So $S_k(f) \rightarrow I_A(f)$. 

3. Let $f = g h_p$ with $g \in \mathcal{C} ([0,1])$ such that $g(0) = g(1) =0$ and $h_p (x)=  x^p (1-x)^p$. We adapt the above argument to show that $S_k (f) \rightarrow I_A(f)$. For any $\ep >0$ there exists a polynomial function $Q$ such that $Q(0) = Q(1) =0$ and $|g - Q| \leqslant \ep$. We deduce that $S_k ( gh_p ) \leqslant S_k ( Q h_p) + \ep S_k ( h_p)$. By Theorem \ref{theo:trace_estim_plus}, $S_k (h_p)$ is bounded above by some constant $C$.    
We can write $Q= Rh$ with $h(x) = x (1-x)$ and $R$ polynomial. So  $Q h_p = Rh_p h$ and part 2 implies that $S_k ( Qh_p) \rightarrow I_A( Q h_p)$. So  $ \limsup S_k ( gh_p) \leqslant I_A (Qh_p) + \ep C.$ Then $I_A(Q h_p) \leqslant I_A ( gh_p) + \ep I_A(h_p)$ so 
$$  \limsup S_k ( f) \leqslant I_A (f) + \ep C + \ep I_A(h_p)$$
Letting $\epsilon \rightarrow 0$, $ \limsup S_k ( f) \leqslant I_A (f)$. By the same argument $\liminf S_k (f) \geqslant I_A(f)$. We conclude that $S_k (f) \rightarrow I_A (f)$. 

4. Let $f$ be a function $f$ satisfying the assumptions of the theorem. Then $f$  has the form $f =g h_{p/2}$ where $g$ is continuous with $g(0) = g(1) = 0$. By part 3, the proof is complete.
\end{proof}

Let us explain how we can deduce the theorems stated in Section \ref{sec:toepl-oper-with}. Theorem \ref{theo:intro_2} contains both Theorem \ref{theo:Weyl_law_singular} and Theorem \ref{theo:pol_convergence}. Theorem \ref{theo:intro} follows from Theorem \ref{theo:Weyl_law_singular} by approximating the characteristic function of $[a,b]$ by continuous functions vanishing on a neighborhood of $0$ and $1$. Corollary \ref{cor:Weyl_law_2_terms} follows from Theorem \ref{theo:Weyl_law_singular} by writing $g(x) = g(0)(1-x) + g(1) x + f(x)$ so that 
$$ \op{tr} (g (T_{A,k})) = g(0) \op{tr} (T_{A^c,k}) + g(1) \op{tr} (T_{A,k}) + \op{tr}( f(T_{A,k}))$$ 
and we estimate the two first traces in the right-hand side with Equation (\ref{dev:f_trace}). 

For any $f \in \mathcal{C}( [0,1], \R)$ and $p \in (0,1]$, let
$$ \| f \|_p := \sup_{t \in [0,1] } \Biggl( \frac{|f(t)|}{ t^p} + \frac{|f(t)|}{(1-t)^p} \Biggr) .$$ 
Observe that $\| f\|_p < \infty$  $\Leftrightarrow$ $f(0) = f(1) =0$ and $f$ is H\"older continuous with exponent $p$ at $0$ and $1$. Observe also that $\| \cdot \|$ is stronger than the uniform norm, that is $\sup|f| \leqslant \| f \|_p$.  
\begin{prop} \label{prop:uniform_result}
For any $p \in (0,1]$, there exists a constant $C>0$ such that for any $f \in \mathcal{C}( [0,1], \R)$, 
\begin{gather} \label{eq:ineq_normp}
 |I(f) | \leqslant C \| f \|_p \qquad \text{ and } \qquad \forall k, \quad |S_k (f) | \leqslant C \| f \|_p.
\end{gather}
For any compact set $K$ of the normed vector space $\bigl( \{ f \in \mathcal{C}([0,1], \R),\; \|f\|_p < \infty \}, \; \| \cdot \|_p\bigr) $,  $S_k (f) \rightarrow ( 2\pi)^n \op{vol} ( \partial A) I(f)$ uniformly with respect to $f \in K$. 
\end{prop}
\begin{proof} Let $h_p(t) = t^p (1-t)^p$. We easily check that
$$ | f(t) | \leqslant 2^p \| f\|_p h_p (t) , \qquad \forall t \in [0,1].
$$
Since $S_k$ and $I$ are positive functionals, this implies that
$$ |S_k (f) | \leqslant 2^p \| f \|_p S_k ( h_p) , \qquad |I(f)| \leqslant 2^p \| f \|_p  I ( h_p).$$
We deduce (\ref{eq:ineq_normp}) by using that $I(h_p) < \infty$ and that the sequence $(S_k (h_p))$ is bounded by Theorem \ref{theo:Weyl_law_singular}. By the same theorem, we already know that $S_k (f) \rightarrow ( 2\pi)^n \op{vol} ( \partial A) I(f)$ when $\| f \|_p < \infty$. Using (\ref{eq:ineq_normp}) with an $\ep/3$-argument, we show that this convergence is uniform on compact sets, as stated in the proposition.
\end{proof}

\section{Quantum probabilities for fermions} \label{sec:quant-prob-ferm}

In this section, we will compute the distribution of $N_A^\Psi$ defined in (\ref{def:N_A_Slater}) and prove the formulas (\ref{eq:entropy}) for the entanglement entropy. We view $\mathcal{H}_k$ as a subspace of the space $\mathcal{E}_k$ of square integrable sections of $L^k$.  The multiplication by the characteristic function of $A$ is an orthogonal projector $P_A$ of $\mathcal{E}_k$ and we have a corresponding decomposition $\mathcal{E}_k = \op{Im} P_A \oplus \op{Ker} P_A$.  Here the space $\mathcal{E}_k$ is infinite dimensional, but in all the considerations to come, we can replace $\op{Im} P_{A}$, $\op{Ker} P_A$ and $\mathcal{E}_k$  by the spaces $P_A ( \mathcal{H}_k)$, $(1-P_A) (\mathcal{H}_k)$ and $P_A( \Hilb_k) \oplus (1- P_A) (\Hilb_k)$.

So we assume that we are in the following situation: $\mathcal{E}$ is a finite dimensional complex Hilbert space, $P_A$ is an orthogonal  projector of $\mathcal{E}$ and  $\mathcal{H}$ is a subspace of $\mathcal{E}$. Let $\mathcal{E}_A = \op{Im} P_A$, $\mathcal{E}_B = \op{Ker} P_A$ so that $\mathcal{E} = \mathcal{E}_A \oplus \mathcal{E}_B$. Let $T_A$ be the Hermitian endomorphism of $\mathcal{H}$ defined by $\langle T_A s, t \rangle = \langle P_A s,t \rangle$ for all $s, t \in \mathcal{H}$. 
Let $s_1$, \ldots $s_d$ be an orthogonal basis of eigenvectors of $T_A$ and set 
$$ \Psi := s_1 \wedge \ldots \wedge s_d \in \textstyle{\bigwedge} \Hilb. $$ 
We denote by $\la_i$ the eigenvalue of $s_i$ and  write $s_i = s_i^A + s_i^B \in \mathcal{E}_A \oplus \mathcal{E}_B$.  
\begin{lemma} \label{lem:calcul}
$ $
\begin{enumerate}
\item  We have $ \langle s_i^A , s_j^A \rangle  = \la_i \delta_{ij}$ and $ \langle s_i^B , s_j^B \rangle  = (1-\la_i) \delta_{ij}$. 
\item The decomposition of $\Psi$ in $\bigwedge \mathcal{E} = \bigl( \bigwedge \mathcal{E}_A \bigr) \otimes \bigl( \bigwedge \mathcal{E}_B \bigr)$ is given by 
\begin{gather} \label{eq:10}
 \Psi = \sum (-1)^{\ep_I} s_I^A \otimes  s_{I^c}^B 
\end{gather}
where we sum over the subsets $I$ of $\{ 1, \ldots, d \}$ and   $\ep_I = \pm 1$, $s_I^{A} = s_{i_1}^A\wedge \ldots \wedge s_{i_m}^A$ and $s_{I}^B =s_{i_1}^B\wedge \ldots \wedge s_{i_m}^B$ if $I$ consists of $i_1 < \ldots < i_m $.
\end{enumerate}
\end{lemma}

\begin{proof} 
We have $\langle s_i^A, s_j^A \rangle = \langle s_i^A, s_j \rangle = \langle T_A s_i , s_j \rangle = \la_{i} \delta_{ij}$. The proof of the second relation is similar.
For the decomposition of $\Psi$, we write 
$$s_i = s_i ^A + s_i^B = s_i^A \otimes 1_{\bigWedge \mathcal{E}_B} + 1_{\bigWedge \mathcal{E}_A} \otimes s_i^B \in  \bigl( \bigWedge \mathcal{E}_A \bigr) \otimes \bigl( \bigWedge \mathcal{E}_B \bigr)$$
and expand the product defining $\Psi$. 
\end{proof}
Equation (\ref{eq:10}) is the Schmidt decomposition of $\Psi$. Indeed, by the first part of Lemma \ref{lem:calcul}, the vectors $s_I^A \otimes s_{I^c}^B$ are mutually orthogonal. Be aware that these vectors are not normalised, actually 
\begin{gather} \label{eq:11}
\| s_I^A \otimes s_{I^c}^B \|^2 =  \la_I (1- \la)_{I^c} 
\end{gather}
where $\la_I = \la_{i_1} \ldots \la_{i_m}$ if $I = \{ i_1, \ldots, i_m \}$.

Consider the decomposition 
\begin{gather} \label{eq:9} 
{\bigWedge} \mathcal{E} = \bigoplus_{\ell} \Bigl( {\bigWedge}^\ell \mathcal{E}_A \Bigr) \otimes \Bigl( \bigWedge \mathcal{E}_B \Bigr)
\end{gather}
Let $N_A$ be the endomorphism of $\bigwedge \mathcal{E}$ equal to $\ell$ in the $\ell$-th summand. Let $N_A^\Psi$ be the probability distribution of $N_A$ in the state $\Psi$. By (\ref{eq:10}) and (\ref{eq:11}),
\begin{gather} \label{eq:8}
 p ( N_A^{\Psi} = \ell ) = \sum _{ |I | = \ell } \la_I (1- \la)_{I^c} 
\end{gather}
Here we recognize the distribution of a sum of independent Bernoulli random variables with parameters $\la_i$, $i =1, \ldots , d$. Indeed, consider the product probability measure of $\{ 0, 1\}^d$ given by $\Pr ( \ep ) = \prod \Pr_i ( \ep_i)$ where $\Pr_i ( 1) = \la_i$ and $\Pr_i( 0 ) = 1 - \la_i$. Then identifying the set of subsets of $\{1, \dots , d \}$ with $\{ 0,1 \}^d$ by sending $\ep$ into $I = \{ i / \ep_i =1 \}$, we see from (\ref{eq:8}) that $N_A^{\Psi}$ has the same distribution as $\sum \ep_i$.

Observe that the probability that $s_i$ belongs to $A$ is  $ \| s_i \|_A^2 = \langle P_A s_i, s_i \rangle = \langle T_A s_i, s_i \rangle = \la_i$. So for what concerns the number of particles in $A$, the Fermionic state $\Psi$ behaves as $d$ independent 1-particle states $s_1$, \ldots, $s_{d}$.

%% \begin{prop}
%% The expectation and variance of $ N_A^{\Psi}$ are given by 
%% $$ \mathbb{E} ( N_A^\Psi ) = \op{tr} (T_A), \qquad  \mathbb{E}\bigl( (N_A^\Psi)^2\bigr) - \bigl( \mathbb{E}(N_A^\Psi) \bigr)^2 = \op{tr} (T_A - T_A^2) .$$
%% More generally, the cumulants of $ N_A^{\Psi}$ are
%% \begin{gather} \label{eq:cum_prop}
%%  \ka_{\ell} ( N_A^{\Psi} ) = \op{tr}  P_\ell ( T_{A}), \qquad \ell \in \N^*
%% \end{gather}
%% where the $P_{\ell}$ are the polynomial functions defined by the generating function $ \ln ( 1 + (e^t -1 ) X) = \sum P_\ell (X) t^\ell/\ell!$.
%% \end{prop}

%% \begin{proof} A probabilist will simply use the additivity property of the cumulant. Otherwise we can do a direct computation as follows.
%% For any $x \in \R$, we have by (\ref{eq:8}) that 
%% \begin{xalignat*}{2}
%%  \mathbb{E} ( x^{N_A^\Psi}  ) =  & \sum_{I} x^{|I|} \la_I ( 1 - \la)_{I^c} = \prod_{i=1}^d \bigl( x\la_i + (1-\la_i) \bigr) \\ = & \det ( 1  + (x-1) T_A )
%% \end{xalignat*}
%% Setting $x = e^t$, we get $ \op{ln} \mathbb{E} ( e^{ tN_A^{\psi}} ) = \op{tr} ( \ln ( 1 + ( e^t -1 ) T_A ))$.
%% \end{proof}

Let $\rho \in \mathcal{S} ( \bigwedge \mathcal{E})$ be the orthogonal projector of $\bigwedge \mathcal{E}$ onto the line generated by $\Psi$. Viewing $\bigwedge \mathcal{E}$ as a bipartite system $\bigwedge \mathcal{E}= \bigwedge \mathcal{E}_A \otimes \bigwedge \mathcal{E}_B$, we introduce the reduced state $\rho_A = \op{tr}_{\wedge \mathcal{E}_B} (\rho ) \in \mathcal{S}(  \bigwedge \mathcal{E}_A)$. The {\em entanglement spectrum} of $\Psi$ is by definition the spectrum of $\rho_A$. The  {\em entanglement entropy} of $\Psi$ is $S( \rho_A) = -\op{tr} ( \rho_A \ln \rho_A )$. 

\begin{prop} $ $
\begin{enumerate} 
\item The entanglement spectrum of $\psi$ consists of the $\la_I(1-\la)_{I^c}$ where $I$ runs over the subsets of $\{ 1, \ldots ,d \}$.  
\item $S ( \rho_A) = \op{tr} \bigl( f (T_A)\bigr)$  with $ f (x) = -x \ln (x) - (1-x) \ln (1-x)$.
\end{enumerate}
\end{prop}

\begin{proof} The computation of the spectrum follows from the Schmidt decomposition (\ref{eq:10}) having in mind (\ref{eq:11}). We deduce that  
$$ -S ( \rho_A) = \sum_I \la_I(1-\la)_{I^c} \ln ( \la_I(1-\la)_{I^c}) $$ 
Write $\ln ( \la_I(1-\la)_{I^c}) =\sum_{i \in I} \ln \la_i + \sum_{i \in I^c} \ln ( 1-\la_i)$. Reorganising the sums
 with $\sum_I ( \sum_{i \in I} + \sum_{i \in I^c}) = \sum_i ( \sum_{I \ni i} + \sum_{I^c \ni i } )$, we get 
$$ - S( \rho_A) = \sum_i \Bigl[ ( \ln \la_i) \sum_{I \ni i} \la_I ( 1- \la)_{I^c} + \ln ( 1- \la_i) \sum_{I^c \ni i} \la_I ( 1- \la)_{I^c} \Bigr] .$$
Finally, we claim that the sums over $I$ in the right-hand side are respectively equal to $\la_i$ and $1-\la_i$. For instance, the first sum for $i=1$ is equal to 
\begin{gather*}
 \sum_{I \ni 1}  \la_I ( 1- \la)_{I^c} = \la_1 \sum_{J \subset \{ 2, \ldots, d \}} \la_J ( 1- \la)_{J^c} \\ 
= \la_1 ( \la_2 + ( 1-\la_2))\ldots \bigl( \la_n + (1-\la_n) \bigr) = \la_1,
\end{gather*}
which concludes the proof.
\end{proof}

\section{Applications} 
In this section, we prove the results presented in Section \ref{sec:appl-iqh-stat} of the introduction.
So consider a measurable subset $A$ of $M$ and let $N_A^\Psi$ be the probability distribution of the number of particles in $A$ of the Slater determinant $\Psi$. Denote by $ \la_1 \leqslant \ldots \leqslant \la_{d_k}$ the eigenvalues of the Toeplitz operator of $\mathcal{H}_k$ corresponding to $A$. 
By (\ref{eq:8}), $N_A^\Psi$ has the same distribution as the sum of independent Bernoulli random variable $B( \la_i)$, $i =1, \ldots , d_k$. Denote by $f( \cdot ,\lambda)$ the cumulant generating function of the fluctuation $\widetilde {B} ( \la)  = B (\la) - \mathbb{E} ( B( \la))$
\begin{gather} \label{eq:16}
 f(t,\la) := \ln \mathbb{E} ( e^{ t \widetilde{B}(\la) }) = \ln ( e^{-t \la} ( 1 - \la + \la e^t)).
\end{gather}
We denote also by $\widetilde{N}_A^\Psi$ the fluctuation of $N_A^\Psi$.

\begin{prop} \label{prop:cumulant_generating_function}
Assume that $A$ has a smooth boundary. Then 
$$ \lim_{k \rightarrow \infty} k^{-n + 1/2} \ln \mathbb{E} ( e ^{t \widetilde{N}_A^\Psi }) =  \frac{\op{vol} ( \partial A)}{(2 \pi)^n} I (f( t, \cdot)) $$
in the $\Ci$-topology on any compact subset of $B = \{ t \in \C/ \; | \op{Im} t| < \pi \}$. 
\end{prop}

\begin{proof} 
We use that $\ln \mathbb{E} ( e ^{t \widetilde{N}_A^\Psi }) = \sum_{i=1}^{d_k} \ln \mathbb{E} ( e^{ t \widetilde{B}(\la_i) }) = \op{tr} ( f (t, T_A)) .$
Since $f$ is smooth on $B \times [0,1]$ and $f(t,0) = f(t,1) =0$, the pointwise convergence follows from Theorem \ref{theo:Weyl_law_singular}. To prove the convergence in $\Ci$-topology, we apply Proposition \ref{prop:uniform_result} by using that  $B \rightarrow \Cl^1 ( [0,1])$, $t\rightarrow f(t,\cdot)$ is smooth, and that the $\Cl^1$ norm is stronger than $\| \cdot \|_p$.
 %% Since $f$ is smooth, the map $B \rightarrow \Cl^1 ( [0,1])$, $t\rightarrow f(t,\cdot)$ is smooth. Since the $\Cl^1$ norm is stronger than $\| \cdot \|_p$ and the functional $I$ and $S_k$ are continuous with respect to $\| \cdot \|_p$, we deduce that $t \rightarrow S_k ( f(t, \cdot)$ and $t \rightarrow I(f( t, \cdot))$ are smooth. Furthermore, the derivatives are given by  $D^\ell_t S_k ( f (t, \cdot)) = S_k ( D^\ell_t f (t, \cdot))$ and similarly for $I$. Then using the second assertion of Proposition \ref{prop:uniform_result}, we obtain the $\Ci$-convergence on compact subsets of $B$. 
\end{proof}

Recall that the $\ell$-th cumulant $\ka_\ell (X)$ of a probability distribution $X$ is defined by the generating function $$ \ln ( \mathbb{E}( e^{tX} )) = \sum \ka_\ell (X) t^\ell/\ell!.$$ 
In particular, $\ka_1 (X)$ is the mean of $X$ and $\ka_2(X)$ its variance. Since $\kappa_\ell (X) = \kappa_\ell ( \widetilde X)$ for $\ell \geqslant 1$,  we can deduce from Proposition \ref{prop:cumulant_generating_function} that $ k^{-n+1/2} \ka_\ell (N_A^\Psi)$ has a finite limit when $k\rightarrow \infty$. But we can actually get a complete asymptotic expansion in the following way. Expanding the cumulant generating function of the Bernoulli distribution, we have 
$$ \ln ( \mathbb{E} ( e^{ t B( \la)} ) = \ln ( 1 + (e^t -1 ) \la ) = \sum P_\ell ( \la) \la ^\ell/\ell! ,$$
with $P_{\ell} ( \la ) = \kappa_\ell ( B( \la))$. From this expansion, we recover the classical fact that the $P_\ell$ are polynomial functions and that they satisfy the recurrence relation $ P_{\ell+1} (X) = X (1-X) P'_{\ell} (X)$. We deduce in particular that for $\ell \geqslant 2$, 
\begin{gather} \label{eq:par_rel}
P_{\ell} (X ) = (-1)^{\ell} P_{\ell} (1-X).
\end{gather}
The cumulants being additive for independent random variables, we have
$$ \ka_{\ell} ( N_A^{\Psi} ) = \sum_{i=1}^{d_k} \ka_{\ell} (B(\la_i)) =  \op{tr}  P_\ell ( T_{A}), \qquad \ell \in \N^*$$
We can now deduce Theorem \ref{theo:cumulant_estimate} from Theorem \ref{theo:pol_convergence}. The fact that the leading order coefficient $I(P_\ell)$ vanishes when $\ell$ is odd and that the expansions are in powers of $k^{-1}$ is a consequence of the parity relation (\ref{eq:par_rel}) with the second assertion of Theorem  \ref{theo:pol_convergence}.

Corollary \ref{cor:concentration+CLT} is a consequence of the variance estimate (\ref{eq:var_estim}). The first assertion, the concentration inequality, follows from the Chernoff bound under the form \cite[Theorem 2.1.3]{Tao}. The second assertion, the convergence to the normal distribution, follows from Lindeberg-Feller central limit theorem, cf. \cite[Theorem 7.3.5]{Si}. 

As in Section \ref{sec:appl-iqh-stat}, let $\widetilde{B}(p) = B(p) - p$ and for any $\al >0$, let $X_n (\al)$ be the sum of $2n+1$ independent random variables $\widetilde{B}( \op{er} ( \al m ) )$, $m =-n, -n +1, \ldots, n$. Introduce the same function $f$ as in  (\ref{eq:16}) and the same domain $B$ as in Proposition \ref{prop:cumulant_generating_function}. 
\begin{prop} $X_n (\al)$ converges as $n \rightarrow \infty$ to a random variable $X(\al)$ with vanishing odd cumulants and even ones given by  
\begin{gather} 
 \kappa_{2 \ell} (X(\al)) = \sum_{m \in \Z} P_{2 \ell} ( \op{er} ( \al m) ) = \al^{-1} I ( P_{2 \ell}) + \bigo ( \al^{\infty})  
\end{gather}
in the limit $\al \rightarrow 0$. 
Furthermore,
\begin{gather} \label{eq:17}
 \ln E( e^{t X(\al)}) = \al^{-1}    I ( f( t, \cdot)) + \bigo ( \al^{\infty}) 
\end{gather}
where the $\bigo$ is uniform on compact subsets of $B$.
\end{prop}

\begin{proof} 
Since the random variables $X_n(\al)$ take their values in $\frac{1}{2} \Z $, 
$$ P ( X_n ( \al) = \ell /2 ) = \frac{1}{\pi} \int_{-\pi/2}^{\pi/2} e^{ -it \ell/2} E( e^{it X_n( \al)}) \; dt .$$ 
Let us prove that these integral have a limit when $n \rightarrow \infty$. A straightforward computation leads to 
$$ E( e^{t (\tilde{B}( \la) +\tilde{B}(1-\la) )}) = 1 + 2 \la ( 1 - \la) ( \cosh t -1)$$
Using the relation $\er (x) + \er ( 1-x) =1$, we get that   
$$ E( e^{t X_n ( \al) }) = \cosh ( t/2) \prod_{m=1}^{n} ( 1 + 2 \la_m ( 1- \la_m) ( \cosh t -1)) $$
with $\la_m = \er ( \al m)$. Using that $I (P_2)$ is finite, we obtain that $\sum \la_m ( 1- \la_m) $ is finite, so $E( e^{t X_n ( \al) })$ converges as $n \rightarrow \infty$ to some limit $\varphi (\al)(t)$ uniformly on any compact set of $\C \ni t$. We now can define the random variables $X(\al)$ by 
$$   P ( X ( \al) = \ell /2 ) = \frac{1}{4\pi} \int_{-2\pi}^{2\pi} e^{ -it \ell/2} \varphi ( it)  \; dt , \quad \ell \in \Z. $$
The convergence of the characteristic functions being uniform on compact sets, we certainly have  that  
$$\kappa_{\ell} (X(\al)) = \sum_{m \in \Z} \ka_\ell ( \tilde{B}( \la_m) )  = \sum_{m \in \Z} P_{\ell} ( \op{er} ( \al m) )$$
which is equal to $ \al^{-1} I ( P_{\ell}) + \bigo ( \al^{\infty})$ by Lemma \ref{lem:Euler_Maclaurin}.  The proof of (\ref{eq:17}) is similar. 
\end{proof}

\bibliographystyle{alpha}
\bibliography{biblio}

\bigskip

\noindent
\begin{tabular}{ll}
Laurent Charles & Benoit Estienne \\  Sorbonne Universit\'e, CNRS, \qquad $ $ & Sorbonne Universit\'e, CNRS, \\
 Institut de Math\'{e}matiques  &  Laboratoire de Physique Th\'eorique \\
de Jussieu-Paris Rive Gauche,  &  et Hautes
\'Energies, \\
IMJ-PRG, F-75005 Paris, France. & LPTHE, F-75005 Paris, France. 
\end{tabular}

\end{document}